\documentclass{article}
\usepackage{amsmath}
\usepackage{amssymb}
\usepackage{amsthm}
\usepackage[authoryear]{natbib}
\usepackage{booktabs}
\usepackage{graphicx}
\usepackage{epstopdf}
\usepackage{caption}
\usepackage{bm}
\usepackage{bbm}
\usepackage{xcolor}
\usepackage{pdflscape}
\usepackage{longtable}
\usepackage[referable]{threeparttablex}
\usepackage{mathrsfs}
\usepackage[english]{babel}
\usepackage{subcaption}
\usepackage[a4paper,margin=1.25in]{geometry}
\usepackage[title]{appendix}
\usepackage{float}
\usepackage{multirow}
\usepackage[onehalfspacing]{setspace}

\usepackage[hypertexnames=false]{hyperref}
\definecolor{lightblue}{rgb}{0, 0.4, 0.6}
\hypersetup{
    breaklinks=true,
	colorlinks = true,
	citecolor = {lightblue},
	linkcolor = {lightblue},
	urlcolor = {lightblue}
}

\usepackage[capitalise,nameinlink,noabbrev]{cleveref}
\crefname{equation}{}{}
\crefname{appsec}{Appendix}{Appendices}
\crefname{appsubsec}{Appendix}{Appendices}
\usepackage{autonum} 

\usepackage[inline,shortlabels]{enumitem}
\newlist{assenumerate*}{enumerate*}{4} 
\setlist[assenumerate*]{label=(\roman*), ref=\theassumption.\roman*}
\crefname{assenumerate*i}{Assumption}{Assumptions}

\newlist{assenumerate}{enumerate}{4} 
\setlist[assenumerate]{label=(\roman*), ref=\theassumption.\roman*}
\crefname{assenumeratei}{Assumption}{Assumptions}

\usepackage{tikz}
\usetikzlibrary{shapes}
\usetikzlibrary{decorations.pathreplacing}

\newtheorem{counter}{counter} 
\newtheorem{theorem}[counter]{Theorem}
\crefname{theorem}{Theorem}{Theorems}

\crefname{conjecture}{Conjecture}{Conjectures}
\newtheorem{lemma}[counter]{Lemma}
\crefname{lemma}{Lemma}{Lemmas}

\crefname{definition}{Definition}{Definitions}
\newtheorem{assumption}{Assumption}
\renewcommand{\theassumption}{A\arabic{assumption}} 
\crefname{assumption}{Assumption}{Assumptions}

\crefname{corollary}{Corollary}{Corollaries}
\theoremstyle{definition}

\newcommand{\bs}[1]{\boldsymbol{#1}}

\DeclareMathOperator{\tr}{tr}
\DeclareMathOperator{\E}{E}
\DeclareMathOperator{\var}{var}

\DeclareMathOperator{\lambdamax}{\lambda_{\max}}
\DeclareMathOperator{\lambdamin}{\lambda_{\min}}

\DeclareMathOperator{\vecb}{vecb}

\newcommand{\toprob}{\overset{p}{\rightarrow}}

\newcommand{\R}{\mathbb{R}}
\newcommand{\N}{N}
\newcommand{\1}{\mathbbm{1}}
\newcommand{\nmax}{n_{\max}}
\newcommand{\mumin}{\mu_{\min}}
\DeclareMathOperator{\diag}{diag}
\renewcommand{\th}[1][th]{^{\text{#1}}} 

\newcommand{\citepos}[1]{\citeauthor{#1}'s (\citeyear{#1})} 

\title{Multidimensional clustering in judge designs\footnote{We thank Tom Boot, John Chao and Norman Swanson for insightful comments and discussions.}}
\author{Johannes W. Ligtenberg\\ University of Groningen\\\texttt{j.w.ligtenberg@rug.nl} \and Tiemen Woutersen\\ University of Arizona \\\texttt{woutersen@arizona.edu}}
\date{June 13, 2024}

\begin{document}
\maketitle

\begin{abstract}
    \noindent 
    Estimates in judge designs run the risk of being biased due to the many judge identities that are implicitly or explicitly used as instrumental variables. The usual method to analyse judge designs, via a leave-out mean instrument, eliminates this many instrument bias only in case the data are clustered in at most one dimension. What is left out in the mean defines this clustering dimension. How most judge designs cluster their standard errors, however, implies that there are additional clustering dimensions, which makes that a many instrument bias remains. We propose two estimators that are many instrument bias free, also in multidimensional clustered judge designs. The first generalises the one dimensional cluster jackknife instrumental variable estimator, by removing from this estimator the additional bias terms due to the extra dependence in the data. The second models all but one clustering dimensions by fixed effects and we show how these numerous fixed effects can be removed without introducing extra bias. A Monte-Carlo experiment and the revisitation of two judge designs show the empirical relevance of properly accounting for multidimensional clustering in estimation.
    \\[1ex]
\end{abstract}

\section{Introduction}
Due to the plausibility of their identification strategy, judge designs are popular to identify the effect of a judge's decision on an outcome variable. These studies use that cases are randomly assigned to judges and that the judges differ systematically in their decisions. Which judge presides can therefore serve as an instrumental variable (IV) for its endogenous verdict.

The number of judges and hence the number of IVs, is typically large and thus may cause a many instrument bias. The conventional method to analyse judge designs conceals that many IVs are used because, in a first step, it aggregates the judge identities into a single instrument. In general, however, combining the IVs does not remove the many instrument bias.

Most judge designs aggregate the IVs via a leave-out mean. This specific combination method overcomes the many instrument bias, but it does so only in case the data are clustered in at most one dimension. This dimension is implicitly defined by what is left out in the mean. The vast majority of judge designs, however, cluster their data in different dimensions when estimating the variance. The effect of these additional clustering dimensions on the many instrument bias is not well studied and is therefore the topic of this paper.

To show the empirical relevance of multidimensional clustering, we analyse 65 papers that use a judge design, listed in \citet{chyn2024examiner}. \cref{app:overview} gives the details of the overview to which we will return throughout the paper. For starters, we note that 58 papers use some form of leave-out, which therefore defines a clustering dimension. 51 out of these 58 cluster their standard errors in either the main specification or in the robustness checks on a different dimension, hence implying multidimensional clustering. Furthermore, eleven papers explicitly allow for multidimensional clustered data by using two-way clustered standard errors. Another seven papers use one way clustered standard errors, but are concerned with additional dependence, such that in robustness checks they cluster in other dimensions.

In this paper we study the effect of multidimensional clustering on the many instrument bias in judge designs. We start by analysing the effect of clustering in a single dimension. We show that clustering in the data exacerbates the many instrument bias in the conventional two stage least squares (2SLS) estimator compared to independent data. Moreover, clustering makes that the usual solution to the many instrument bias, the jackknife instrumental variable estimator  \citep[JIVE,][]{angrist1999jackknife,blomquist1999small}, no longer fully removes this bias. In case the data are clustered in a single dimension the cluster adaptation to the JIVE, the CJIVE  \citep{ligtenberg2023inference,frandsen2023cluster}, does remove the many instrument bias.

Furthermore, this CJIVE is equivalent to the usual approach to analyse judge designs which combines the judge dummies into a single instrument via a leave-out mean. This equivalence explains why there is no many instrument bias in judge designs with the data clustered in a single dimension. In addition, the equivalence allows us to focus on the model with an instrument for every judge.

The CJIVE only removes the bias due to clustering in a single dimension. Clustering in additional dimensions adds bias terms to the 2SLS estimator, which are not removed by the CJIVE. We remove these additional terms in a new estimator, the multidimensional (MD) CJIVE, which can be used for many instrument bias free estimation with multidimensional clustered data.

Judge designs often include control variables in their model. Removing these control variables, such that JIVE or one of its variants can be applied, can introduce extra clustered dependence in the data \citep{chao2023jackknife}. In particular, in case the control variables take the form of a fixed effect (FE), the control variables are removed by subtracting group means from the observations. Since the group mean contains all observations from the group, this introduces additional clustering in the data. Our overview of 65 papers shows that FE are commonplace, since all papers use them in some form.

For the MD CJIVE this extra dependence means that additional terms need to be removed to eliminate the bias, but poses no further problems. Another way to account for the extra dependence caused by removing the FE is via the FE JIVE from \citet{chao2023jackknife}. This estimator adapts the JIVE to correctly remove the FE in case the data are independent. The FE JIVE is therefore an alternative to the MD CJIVE in case all clustered dependence can be modelled by including FE. Since it is not always plausible that the clustering in every dimension can be captured by FE, we note that the FE CJIVE extension of the FE JIVE from \citet{ligtenberg2023inference} allows for more general dependence in one dimension. 

We briefly compare the MD CJIVE and the FE CJIVE. We find that the FE CJIVE may work well even if the clustering cannot be fully captured by FE. This happens when the elements of the instrument projection matrix which are used to weight the bias terms, do not covary systematically with the covariances of the errors. Furthermore, we show that in some models the more general MD CJIVE does not identify the parameter of interest, whereas the FE CJIVE does.

Our overview shows that the judge level is a particularly popular clustering dimension. Out of the 65 papers 30 cluster on the judge in the main specification. Another 8 papers do so in robustness checks. Clustering on this level, suggests that after accounting for all variables in the model, there remains dependence between the errors of the cases handled by the same judge. We show that this type of dependence is especially heavily weighted in the bias terms, but that in general neither the MD CJIVE nor the FE CJIVE can correct for clustering on the judge level.  

We compare the different estimators in a Monte-Carlo simulation where we generate data with clustering in multiple dimensions. We find that 2SLS and JIVE are biased in this case. The bias of the CJIVE is close to zero in some simulation set ups, but larger in others. Furthermore, the FE (C)JIVE can be biased in case general clustering is modelled with FE. The estimates of the MD CJIVE perform especially well in case there is complex clustering on both dimensions.

Finally, we revisit two judge designs to study the effect of clustering on the many instrument bias in the coefficient estimates. The first is by \citet{di2013criminal} who estimate the effect of electronic monitoring (EM) as an alternative to incarceration on recidivism. In this study there are possibly three more clustering levels beyond the individual level that is implied by the leave-out mean instrument. When we take these dimensions into account with the MD CJIVE, we find that the effect of EM on recidivism is slightly more negative than when we ignore the additional clustering. The FE CJIVE yields no interpretable results.

The second study we revisit is the one by \citet{agan2023misdemeanor}. These authors estimate the effect of nonprosecution on recidivism. Due to FE there may be two additional clustering dimensions that are not taken into account in estimation. The MD CJIVE that removes from the estimator the terms associated to these dimensions, yields estimates that are less negative than when these dimensions are ignored. The FE CJIVE estimates are again less stable, but if we ignore the one outlier, we find that the estimates that correctly incorporate the control variables are more negative than those that naively remove the control variables.

\paragraph{Literature review}
That many instruments bias conventional estimators has long been acknowledged. \citet{angrist1999jackknife} for example used the theory from \citet{nagar1959bias} to show how the bias of the 2SLS estimator increases in the number of instruments. \citet{angrist1999jackknife} and \citet{blomquist1999small} propose JIVEs that are robust against many instruments as they remove the terms in 2SLS that cause the bias. We take a similar approach for the MD CJIVE.

Also for judge designs a many instrument bias has been a concern from the very beginning. \citet{kling2006incarceration}, who initiated the judge design literature, used JIVE on the separate judge identity dummies. In later studies it has been more common to use a leave-out mean instrument and 2SLS instead of JIVE, as also comes forth from our overview in \cref{app:overview}. Although with this estimation method it is not immediately clear that many instruments are used, the topic kept resurfacing in the judge design literature. See for example the notes in \citet{hull2017examiner} and \citet{mikusheva2021inference}. 

The jackknife that is used in the JIVE is more broadly applicable to overcome many instrument problems in IV models. For instance, it has been applied to adapt LIML and Fuller type estimators in \citet{hausman2012instrumental} and to derive identification robust tests in \citet{crudu2021inference}, \citet{mikusheva2021inference} and \citet{matsushita2020jackknife}. In the latter setting the jackknife for independent data was generalised to allow for one-way clustered data by \citet{ligtenberg2023inference}. Subsequently, \citet{frandsen2023cluster} used this cluster jackknife in a cluster JIVE.

The cluster jackknife allows for data clustered in only a single dimension. Ever since \citet{cameron2011robust}, however, two- or multi-way clustered standard errors are commonplace. In this paper we therefore generalise the cluster jackknife to a MD cluster jackknife. We furthermore note that the focus of the paper by \citet{cameron2011robust} and the clustering literature in general, has mainly been on the effect of clustering on variance estimators and inference, whereas this paper studies the effect of clustering on bias in estimators for the coefficients.

A discussion on clustering specified to judge designs followed from the discussion on clustering in the treatment effect literature. There they showed that which clustering dimensions matter depends on which view one takes on the data generating process (DGP). In particular, they discern a sampling based view in which the data are drawn from an infinitely sized super population, and a design based view where a large part of the population is observed, but where the assignment of treatment differs over samples \citep{abadie2020sampling,abadie2023should}. \citet{chyn2024examiner} took the latter perspective for judge designs. They argued that the observations should be clustered, both for estimation and for inference, in how they are assigned to the judges. Here we abstract from the different views on the DGP and we take the clustering dimensions as given.

\paragraph{Outline}
In \cref{sec:model} we present the general judge design and the effects of clustering on 2SLS, JIVE and CJIVE. This section also shows the equivalence of CJIVE and 2SLS with a leave-out instrument. \Cref{sec:MDC} discusses multi-dimensional clustering and the MD CJIVE. The inclusion of FE and how the FE (C)JIVE can be an alternative to the MD CJIVE are covered in \cref{sec:FE}. In the section that follows we compare the two. \Cref{sec:clustering on judge} adds a note on clustering at the judge level. We assess the  performance of the different estimators in the Monte-Carlo simulations in \cref{sec:MC} and the empirical applications in \cref{sec:applications}. \Cref{sec:conclusion} concludes.

\paragraph{Notation}
Throughout we use the following notation. For any matrix $\bs A$, define the projection matrices $\bs P_{A}=\bs A(\bs A'\bs A)^{-1}\bs A'$ and $\bs M_{A}=\bs I-\bs P_{A}$. Let $\odot$ be the Hadamard product. Write $\bs D_{A}=\bs I\odot\bs A$ if $\bs A$ is a square matrix and $\bs D_{v}$ the diagonal matrix with the elements of $\bs v$ on its diagonal if $\bs v$ is a vector. Then define $\dot{\bs A}=\bs A-\bs D_{A}$ as the matrix with the diagonal elements set to zero. Finally, $\sum_{g\neq h}$ can be the double sum $\sum_{g=1}^G\sum_{h=1,h\neq g}^H$, whether this is the case and the values of $G$ and $H$ depend implicitly on the context. Other sums are defined similarly.

\section{Model and clustering}\label{sec:model}
In the general judge design, one is interested in the effect of a decision taken by an expert on an outcome variable. The expert and the decision can vary from judges needing to decide on conviction to patent examiners
needing to decide on whether a patent is granted. See \citet{chyn2024examiner} for an overview. For ease of exposition we stay close to the original setting of \citet{kling2006incarceration} and call the expert a judge, what the judge needs to decide on a case and a positive decision a conviction.

The judge may be influenced by factors unobservable to the researcher that affect both his decision and the outcome variable. This makes the decision endogenous and consequently the effect of the decision cannot be estimated with ordinary least squares (OLS). Judge designs therefore take a different approach and leverage that judges differ systematically in their conviction rates. That is, some judges are stricter than others. Furthermore, judges are often randomly assigned to cases. These two points combined make that the identity of the judge is a valid and relevant IV for the decision on the outcome variable.

To be precise, we consider the following model to estimate the coefficient of interest $\beta$
\begin{equation}\label{eq:judge IV}
    \begin{split}
        y_i&=X_i\beta+\varepsilon_i\\
        X_i&=\bs Z_i'\bs\Pi+\eta_i,
    \end{split}
\end{equation}
where for case $i=1,\dots,n$, $y_i\in\R$ is the outcome variable, $X_i\in\R$ is a dummy variable indicating conviction, $\bs Z_i\in\R^k$ are mutually exclusive dummy variables indicating which of the $k$ judges handles the case and $\varepsilon_i\in\R$ and $\eta_i\in\R$ are the second and first stage errors. We stack the instruments per case in $\bs Z=(\bs Z_1,\dots,\bs Z_n)'$ and similarly for the other variables. For now we assume there are no control variables, but we come back to that in \cref{sec:FE}.

Since each judge can handle only a limited amount of cases, we allow $k$ to grow with $n$ in the asymptotic approximations. We leave the dependence of $k$, $\bs\Pi$ and $\bs Z$ on $n$ implicit. We make the following assumptions on the judge identity dummies to ensure they are valid and relevant instruments. 
\begin{assumption}
    \begin{assenumerate*}
        \item $\E(\varepsilon_i|\bs Z)=0$ and $\E(\eta_i|\bs Z)=0$ for all $i$.
        \item $\bs\Pi\neq\bs 0$.
    \end{assenumerate*}
\end{assumption}

For now the data can be clustered on a single dimension, which we generalise to more dimensions later in the paper. That the data are clustered means that the data can be partitioned in groups or clusters, and that the observations within a cluster can have any dependence between them, whereas observations from different clusters are independent. We formalise the notion of clustering in the next assumption for which we use the notation from \citet{ligtenberg2023inference} and that we repeat below. 

The $n$ observations can be grouped in $G$ clusters with $n_g$, $g=1,\dots,G$, observations per cluster. Denote the indices of the observations in cluster $g$ by $[g]$ and for any $n\times m$ matrix $\bs A$ let $\bs A_{[g]}$ be the $n_g\times m$ matrix with only the rows indexed by $[g]$ selected. For later purposes, if $\bs A$ is $n\times n$ we also define $\bs A_{[g,h]}$ to be the $n_g\times n_h$ submatrix with the rows in $[g]$ and the columns in $[h]$ selected.

\begin{assumption}
    \begin{assenumerate*}
        Conditional on $\bs Z$, $\{\bs\varepsilon_{[g]},\bs\eta_{[g]}\}_{g=1}^G$ is independent with mean zero.
    \end{assenumerate*}
\end{assumption}

In IV models, one usually estimates $\beta$ via 2SLS. When the number of instruments is large, 2SLS is inconsistent and there is a so-called many instrument bias. To see where this bias comes from and why clustered data can increase this bias relative to independent data, write the 2SLS estimator as $\hat{\beta}=(\bs X'\bs P_{Z}\bs X)^{-1}\bs X'\bs P_{Z}\bs y=\beta+(\bs X'\bs P_{Z}\bs X)^{-1}\bs X'\bs P_{Z}\bs\varepsilon$. Assume that a law of large numbers works on both the part in the inverse and the part after it. Then we can write the latter as
\begin{equation}\label{eq:bias}
    \begin{split}
        \E(\bs X\bs P_{Z}\bs\varepsilon|\bs Z)=\E(\bs\Pi'\bs Z'\bs P_{Z}\bs\varepsilon+\bs\eta'\bs P_{Z}\bs\varepsilon|\bs Z)=\E(\bs\eta'\bs P_{Z}\bs\varepsilon|\bs Z)=\E(\sum_{g=1}^G\sum_{i,j\in[g]}\eta_iP_{Z,ij}\varepsilon_j|\bs Z).
    \end{split}
\end{equation}

Since $\bs Z_i$ are judge dummies, we have an explicit expression for $\bs P_{Z}$. See \cref{app:equivalence} for more details.
\begin{equation}\label{eq:Pz}
    \begin{split}
        P_{Z,ij}&=\begin{cases}
            1/n_{J(i)} &\text{ if } J(i)=J(j)\\
            0 &\text{ otherwise},
        \end{cases}
    \end{split}
\end{equation}
where $J(i)$ is the judge that handles case $i=1,\dots,n$ and $n_{j}$ is the number of cases that judge $j=1,\dots,k$ oversees. This shows that the elements of $\bs P_{Z}$ are well above zero, unless each judge handles an infinite amount of cases.

Therefore $\cref{eq:bias}$ becomes
\begin{equation}\label{eq:bias decompose}
    \begin{split}
        \E(\sum_{g=1}^G\sum_{i,j\in[g]}\eta_iP_{Z,ij}\varepsilon_j|\bs Z)&=\E(\sum_{g=1}^G\sum_{i,j\in[g]}\frac{1}{n_{J(i)}}\1\{J(i)=J(j)\}\eta_i\varepsilon_j|\bs Z)\\
        &=\sum_{i=1}^n\frac{1}{n_{J(i)}}\E(\eta_i\varepsilon_i|\bs Z)+\sum_{g=1}^G\sum_{\substack{i,j\in[g]\\i\neq j}}\frac{1}{n_{J(i)}}\1\{J(i)=J(j)\}\E(\eta_i\varepsilon_j|\bs Z),
    \end{split}
\end{equation}
where $\1\{\cdot\}$ is the indicator function. The first term captures the usual many instrument bias that is also present in case the observations are independent. The second term is the additional bias to due to clustering in the data.

\cref{eq:bias decompose} suggests that a solution to the many instrument bias is to set the elements in $\bs P_{Z}$ corresponding to the two terms to zero. The JIVE sets the elements corresponding to the first term to zero, which yields $\hat{\beta}_{J}=(\bs X'\dot{\bs P}_{Z}\bs X)^{-1}\bs X'\dot{\bs P}_{Z}\bs y$.

Clearly, in case the data are clustered JIVE does not remove all bias terms. This is solved by the CJIVE, which also sets the elements of $\bs P_Z$ corresponding to the second term in \cref{eq:bias decompose} to zero. We write this as $\hat{\beta}_{CJ}=(\bs X'\ddot{\bs P}_{Z}\bs X)^{-1}\bs X'\ddot{\bs P}_{Z}\bs y$, where for any square matrix $\bs A$, $\ddot{\bs A}$ is the matrix with in the rows indexed by $[g]$ the elements in the columns indexed by $[g]$ set to zero, $g=1,\dots,G$. For instance, if the observations are sorted by cluster, the block diagonal elements of $\bs P_Z$ are set to zero.

In this section we have focused on the judge design that estimates $\beta$ with the judge identity dummies as instruments. An alternative, more common way of estimating $\beta$ is via 2SLS with a leave-out mean conviction rate as instrument. The leave-out mean conviction rate is a single instrument with as $i\th$ value the average conviction rate of judge $J(i)$, excluding any cases from the cluster that $i$ belongs to. For example, in judicial contexts it is common to exclude any case that involves the defendant of case $i$, as these cases are expected to be dependent. Therefore, the leave-out mean models the data as clustered on the defendant.

In \cref{app:equivalence} we show that CJIVE and 2SLS with a leave-out conviction rate are equivalent up to weighting. This result was anticipated by  \citet{french2014effect}, \citet{aizer2015juvenile}, \citet{dobbie2017consumer}, \citet{doyle2017evaluating},  \citet{arnold2018racial}, \citet{dobbie2018effects}, \citet{ribeiro2019pretrial}, \citet{hyman2018can}, \citet{baron2022there} and \cite{chyn2024examiner}, among others. This also explains why few authors worry about a many instrument bias in judge designs. Given this equivalence result, we focus also in what follows on the IV model with judge dummies as instruments.

\section{Multidimensional clustering}\label{sec:MDC}
Studies that use a leave-out mean conviction rate as instrument implicitly model the data as clustered on the dimension that is left out, often cases or individuals, and remove the many instrument bias in the estimators due to the dependence in this dimension. The overview of judge designs shows, however, that there are often more dimensions on which the data are dependent.

These judge designs adjust their standard errors to the additional dependence, but leave the point estimators unchanged. Similar to when clustering in a single dimension is ignored, this can lead to a many instrument bias. To formalise this, we introduce the following notation.

The data are clustered on $C$ dimensions with on dimension $c=1,\dots, C$, $G^{(c)}$ groups or clusters. On dimension $c$ there are $n_{g}^{(c)}$ observations in cluster $g=1,\dots,G^{(c)}$. The number of cases handled by judge $j=1,\dots,k$ that are in cluster $g$ on dimension $c$ is given by $n_{j,g}^{(c)}$. We denote the cluster that case $i=1,\dots,n$ is in on dimension $c$ as $G^{(c)}(i)$. Define $[g]^{(c)}$ as indices of the cases in cluster $g$ on dimension $c$. Take $[G(i)]=\bigcup_{c=1}^C[G^{(c)}(i)]^{(c)}$ as the indices of cases that share a cluster with case $i$ in some dimension. 

Then multidimensional clustering is specified by the following assumption.

\begin{assumption}
    Conditional on $\bs Z$ and for any $i=1,\dots, n$ and $j\notin [G(i)]$, $\{\varepsilon_{i},\eta_{i}\}$ is mean zero and independent from $\{\varepsilon_{j},\eta_{j}\}$.
\end{assumption}

For expositions purposes we will focus in the remainder of this section on the case in which there are $C=2$ clustering dimension. Our results extend straightforwardly to more dimensions.

We can then generalise the bias formula given in \cref{eq:bias decompose} as 
\begin{equation}\label{eq:bias decompose md}
    \begin{split}
        \E(\bs\eta'\bs P_{Z}\bs\varepsilon|\bs Z)
        &=\E(\sum_{i=1}^n\frac{\eta_i\varepsilon_{i}}{n_{J(i)}}+\sum_{i=1}^n\sum_{j\in[G^{(1)}(i)]^{(1)}\setminus\{i\}}\1\{J(i)=J(j)\}\frac{\eta_i\varepsilon_j}{n_{J(i)}}\\
        &\quad+\sum_{i=1}^n\sum_{j\in[G^{(2)}(i)]^{(2)}\setminus[G^{(1)}(i)]^{(1)}\setminus\{i\}}\1\{J(i)=J(j)\}\frac{\eta_i\varepsilon_j}{n_{J(i)}}|\bs Z)
    \end{split}
\end{equation}
Compared to \cref{eq:bias decompose} there is now a third term that captures the many instrument bias due to the dependence in the second clustering dimension. This additional term is not removed by a CJIVE that clusters on the first dimension.

We remove the additional bias term in a new estimator called the MD CJIVE. Similar to the CJIVE, the MD CJIVE uses a projection matrix where the elements corresponding to the bias terms are set to zero. To write this estimator succinctly, we define $\bs S^{(c)}$ to be an $n\times n$ selection matrix with in row $i$ zeroes everywhere except in the columns indexed by $[G^{(c)}(i)]^{(c)}$. Also, let $\bs S^{(1,2)}=\bs S^{(1)}\odot\bs S^{(2)}$ be the selection matrix with in row $i$ ones in the columns corresponding to $[G(i)]$. For any $n\times n$ matrix $\bs A$ we write $\bs B^{(c)}_{A}=\bs S^{(c)}\odot\bs A$ and $\bs B^{(1,2)}_{A}=\bs S^{(1,2)}\odot\bs A$. Then $\dddot{\bs P}_{Z}=\bs P_{Z}-\bs B^{(1)}_{P_Z}-\bs B^{(2)}_{P_Z}+\bs B^{(1,2)}_{P_Z}$ and the MD CJIVE $\hat{\beta}_{MDCJ}=(\bs X'\dddot{\bs P}_{Z}\bs X)^{-1}\bs X'\dddot{\bs P}_{Z}\bs y$. The terms that are subtracted remove the elements of $\bs P_{Z}$ causing a bias in both clustering dimensions. However, since the cases that are in the same cluster on both dimensions are subtracted twice, we need to add back these elements, which we do with the final term.

In \cref{app:consistency MD CJIVE} we show consistency of the MD CJIVE under high-level assumptions similar to \citet{frandsen2023cluster}. In \cref{app:variance MD CJIVE} we also present a variance estimator for $\hat{\beta}_{MDCJ}-\beta$, although we make no claims about asymptotic normality. 

\section{Fixed effects}\label{sec:FE}
So far we have considered a judge design model without exogenous control variables. In practice however, most studies include control variables in their model. The overview of judge designs showed that a particular popular type of control variable are FE, for example to control for time or location effects. A model with control variables is
\begin{equation}\label{eq:judge IV exogenous}
    \begin{split}
        \bs y&=\bs X\beta+\bs W\bs\Gamma_2+\bs\varepsilon\\
        \bs X&=\bs Z\bs\Pi+\bs W\bs\Gamma_1+\bs\eta,
    \end{split}
\end{equation}
where $\bs W\in\R^{n\times l}$ are the exogenous control variables and the other variables are as before. Similar to the number of instruments, the number of control variables, $l$, can increase with the sample size, but we again suppress dependence of $l$, $\bs W$, $\bs\Gamma_1$ and $\bs\Gamma_2$ on $n$.

Motivated by the Frisch-Waugh-Lowell theorem, the usual way to control for the exogenous variables $\bs W$, is to project them out. This is done by pre-multiplying both equations in \cref{eq:judge IV exogenous} by $\bs M_{W}$. Then conventional methods such as 2SLS or (C)JIVE, can estimate $\beta$ from this transformed model.

Projecting out a small number of exogenous control variables is in most cases an innocent transformation. When the control variables are numerous on the other hand, as is often the case with FE, projecting them out introduces extra dependence in the data. For FE this dependence takes the form of extra clustering dimensions, since the FE are removed by subtracting the group mean from every observation. If there are many FE, the group mean is taken over a small number of observations, making the contribution of each relatively large. Consequently, all observations within a group have nonnegligible dependence.

For the MD CJIVE the extra within group dependence induced by FE, is no obstacle to estimation without many instrument bias. It only adds an additional clustering dimension, with the clusters defined by the groups.

Ignoring the additional dependence in the transformed model on the other hand, can aggravate or introduce a many instrument bias, as \citet{chao2023jackknife} show. They consider the JIVE for a IV model with many instruments, FE and independent observations. To see why projecting out the FE can lead to a many instrument bias, write JIVE for the transformed model as
\begin{equation}
    \begin{split}
        \hat{\beta}_{J}=(\bs X'\bs M_{W}\dot{\bs P}_{M_{W}Z}\bs M_{W}\bs X)^{-1}\bs X'\bs M_{W}\dot{\bs P}_{M_{W}Z}\bs M_{W}\bs y.
    \end{split}
\end{equation}
For it to have no many instrument bias, we require $\E(\bs\eta'\bs M_{W}\dot{\bs P}_{\bs M_{W}Z}\bs M_{W}\bs\varepsilon|\bs Z)$ to be zero. Although $\dot{\bs P}_{\bs M_{W}Z}$ has zero diagonal, there is no guarantee that $\bs M_{W}\dot{\bs P}_{\bs M_{W}Z}\bs M_{W}$ has. Therefore the products of $\eta_i$ and $\varepsilon_i$ may have nonzero weight. Alternatively, $\{[\bs M_{W}\bs\eta]_{i},[\bs M_{W}\bs\varepsilon]_{i}\}_{i=1}^n$ is not a sequence of conditionally independent variables, which yields a many instrument bias as discussed in \cref{sec:model}.

\citet{chao2023jackknife} provide an alternative JIVE, called the FE JIVE, that ensures that there is no many instrument bias, also after projecting out the control variables. They achieve this by solving for a diagonal matrix $\bs D_{\vartheta}$ which is such that $\bs P_{M_WZ}-\bs M_{W,Z}\bs D_{\vartheta}\bs M_{W,Z}$ has zero diagonal. The solution, $\bs D_{\vartheta}$ is given by the diagonal matrix with on its diagonal $\bs\vartheta=(\bs M_{W,Z}\odot\bs M_{W,Z})^{-1}\diag(\bs P_{M_WZ})$, where $\diag(\bs A)$ is a vector with the diagonal elements of $\bs A$. Using this adapted projection matrix instead of $\dot{\bs P}_{Z}$ in the FE JIVE yields $\hat{\beta}_{FEJ}=(\bs X'(\bs P_{M_WZ}-\bs M_{W,Z}\bs D_{\vartheta}\bs M_{W,Z})\bs X)^{-1}\bs X(\bs P_{M_WZ}-\bs M_{W,Z}\bs D_{\vartheta}\bs M_{W,Z})\bs y$. The bias term of the FE JIVE is given by $\E(\bs\eta'\bs M_{W}(\bs P_{M_WZ}-\bs M_{W,Z}\bs D_{\vartheta}\bs M_{W,Z})\bs M_{W}\bs\varepsilon|\bs Z)=\E(\bs\eta'(\bs P_{M_WZ}-\bs M_{W,Z}\bs D_{\vartheta}\bs M_{W,Z})\bs\varepsilon|\bs Z)=0$ by its zero diagonal, provided the observations are independent.

The FE JIVE is an alternative for the MD CJIVE to estimate $\beta$ many instrument bias free when there is clustering in multiple dimensions. Namely, if we assume that all clustering can be captured by cluster FE, such that after including these FE $\{\varepsilon_i,\eta_i\}_{i=1}^n$ in \cref{eq:judge IV exogenous} is an independent sequence, we can readily apply \citepos{chao2023jackknife} FE JIVE.

The assumption that the clustering in every dimension can be fully captured by FE, is not always plausible. In such cases, the FE CJIVE generalisation of the FE JIVE from \citet{ligtenberg2023inference} may be applicable. This generalisation allows for more general clustering in one dimension. We repeat the FE CJIVE in slightly adapted form below, for which we use the following notation. For any matrix $\bs A$, let $\vecb(\bs A)$ the vectorisation of the elements in the $\bs A_{[g,g]}$, $g=1,\dots,G$ and $\vecb^{-1}(\cdot)$ the operator that constructs out of a vector a matrix such that $\vecb^{-1}(\vecb(\bs A))=\bs S\odot\bs A$, where $\bs S$ is a selection matrix as before, but with the clustering dimension suppressed. Finally, let $*$ be the Kathri-Rao or blockwise Kronecker product, where we take the blocks to be defined by the clusters.

Assume that there is a single clustering dimension and that the observations are ordered per cluster. Then $\bs P_{M_WZ}-\bs M_{W,Z}\bs H\bs M_{W,Z}$ with  $\bs H=\vecb^{-1}((\bs M_{W,Z}*\bs M_{W,Z})^{-1}\vecb(\bs P_{M_{W}Z}))$ has a zero block diagonal and therefore can be used in a FE CJIVE. That is $\hat{\beta}_{FECJ}=(\bs X'(\bs P_{M_WZ}-\bs M_{W,Z}\bs H\bs M_{W,Z})\bs X)^{-1}\bs X'(\bs P_{M_WZ}-\bs M_{W,Z}\bs H\bs M_{W,Z})\bs y$ has no many instrument bias.

In \cref{app:consistency FE CJIVE} we show consistency of the FE CJIVE under assumptions similar to \citet{chao2023jackknife}. We also provide a variance estimator for $\hat{\beta}_{FECJ}-\beta$ in \cref{app:variance FE CJIVE}, but again we make no claims about asymptotic normality.

\section{Comparison of MD CJIVE and FE CJIVE}
In this section we compare the MD CJIVE and FE CJIVE. We first show why the FE CJIVE may work well even if there is general clustering that cannot be captured by a FE. Next, we discuss a case in which the MD CJIVE does not identify the parameter of interest, whereas a FE CJIVE does.

\subsection{Modeling general clustering with a fixed effect}\label{sec:FE for general}
In some cases including cluster FE will eliminate most or even all of the many instrument bias, also when there is more general clustering in the data. Namely, including a cluster FE will set the average covariance within a cluster to zero. Consequently, if the elements of the jackknifed projection matrix that weigh these covariances in the bias do not vary too much or not systematically with the covariances, the weighted average of the covariances will also be close to zero, causing no or little bias.

To formalise this argument, consider the FE JIVE with clustering on a single dimension which is modeled with a FE. As before, let $\E(\bs\eta\bs\varepsilon'|\bs Z)=\bs\Xi$. Also, for notational convenience write $\tilde{\bs P}_{Z}$ for the FE JIVE projection matrix.  The term causing the many instrument bias can be written as
\begin{equation}
    \begin{split}
        \E(\bs\eta'\tilde{\bs P}_{Z}\bs\varepsilon|\bs Z)&=\E(\tr(\tilde{\bs P}_{Z}\bs\varepsilon\bs\eta')|\bs Z)=\tr(\tilde{\bs P}_{Z}\bs\Xi')\\
        &
        =\tr(\tilde{\bs P}_{Z}\dot{\bs\Xi}')=\tr(\tilde{\bs P}_{Z}\bs M_{W}\dot{\bs\Xi}'\bs M_{W})=\sum_{g=1}^G\sum_{\substack{i,j\in[g]\\i\neq j}}[\bs M_{W}\dot{\bs\Xi}'\bs M_{W}]_{ji}\tilde{P}_{Z,ij},
    \end{split}
\end{equation}
where we used that $\tilde{\bs P}_{Z}$ has a zero diagonal. Then, $\bs M_{W}\dot{\bs\Xi}'\bs M_{W}=\dot{\bs\Xi}'-\bs P_{W}\dot{\bs\Xi}'-\dot{\bs\Xi}'\bs P_{W}+\bs P_{W}\dot{\bs\Xi}'\bs P_{W}$ is the matrix of covariances minus the cluster row and cluster column averages, corrected for over counting by adding the cluster row-column averages, which makes the cluster averages of the elements in $\bs M_{W}\dot{\bs\Xi}'\bs M_{W}$ zero. Consequently, if the $\tilde{P}_{Z,ij}$ do not differ systematically within a cluster, the bias will be zero.

An example when this happens is when the elements of $\tilde{\bs P}_{Z}$ are independent of those in $\dot{\bs\Xi}$ and have constant within cluster unconditional expectation $\mu_{g}$. Then we can write
\begin{equation}
    \begin{split}
        \E(\bs\eta'\tilde{\bs P}_{Z}\bs\varepsilon)&=\E(\sum_{g=1}^G\sum_{\substack{i,j\in[g]\\i\neq j}}[\bs M_{W}\dot{\bs\Xi}'\bs M_{W}]_{ji}\tilde{P}_{Z,ij})=\E(\sum_{g=1}^G\sum_{\substack{i,j\in[g]\\i\neq j}}\E([\bs M_{W}\dot{\bs\Xi}'\bs M_{W}]_{ji})\E(\tilde{P}_{Z,ij})\\
        &=\sum_{g=1}^G\mu_{g}n_g(n_g-1)
        \E(\sum_{\substack{i,j\in[g]\\i\neq j}}[\bs M_{W}\dot{\bs\Xi}'\bs M_{W}]_{ji})=0.
    \end{split}
\end{equation}

This result is similar to the one in \citet{chetverikov2023standard}, which states that with an exogenously assigned covariate of interest, it is often not necessary to control for clustering if the model includes an intercept. The argument is that after projecting out the intercept, the exogenous regressor is mean zero. Then, given that the regressor is independent from the errors, the expectation of the product of the errors and the regressor can be separated making the entire expectation zero, which removes the need to account for clustering.

The argument here is similar in the sense that we can split the expectation because one factor in the product is independent form the other. Another similarity is that one of these expectations is zero, which simplifies the expression. The argument is different in the sense that we consider a bias stemming from covariances, rather than a variance estimator itself and, more importantly, the inclusion of an intercept or FE does not make the exogenously assigned part, the $\tilde{\bs P}_{Z}$, mean zero in general. Rather, the covariances are made zero by the FE.

\subsection{When FE CJIVE works, but MD CJIVE does not}
The MD CJIVE can handle more general types of clustering than FE CJIVE. This robustness comes at a price. There are cases in which the MD CJIVE cannot identify the parameter of interest, whereas the FE CJIVE can. The following example illustrates our point.

Suppose that, similar to the judge instruments, the data can be clustered and there is one instrument per cluster. That is, $\bs Z_i$ is a zero vector except for the entry corresponding to the cluster that observation $i$ is in. Unlike the judge instrument however, we assume that there is variation of the instrument within the cluster. Furthermore, we assume that the errors consist of a cluster common component and an idiosyncratic component, such that all clustering can be modelled with a FE.

Since there is variation of the instrument within clusters, a cluster FE does not take away all instrument variation. Nevertheless, by the assumptions on the errors, it adequately models the clustering. A FE CJIVE therefore can identify the parameter of interest.

A MD CJIVE on the other hand, does not work. This is because the instrument vector for each individual contains only one nonzero element corresponding to the entry of the cluster it is in. Assuming the observations are sorted on clusters, this structure makes $\bs P_{Z}$ a block diagonal matrix, with the blocks corresponding to the clusters. A MD CJIVE takes sets all these blocks to zero, such that $\dddot{\bs P}_{Z}$ is a zero matrix, leading to no identification of the parameter of interest.

\section{Clustering on the judge}\label{sec:clustering on judge}
Before applying the MD CJIVE and FE CJIVE to simulated and real data in the next sections, we want to add a note on clustering on a particular dimension: the judge level. We observe that many judge designs cluster their standard errors at the judge level, to allow for covariance in the errors of cases handled by the same judge. For illustration, out of the 65 judge designs we reviewed 30 clustered the standard errors in the main specification on the judge level. Another 8 papers did so in one of the robustness checks.

From \cref{eq:bias decompose md} we know that dependencies in the errors can bias the estimates. For clustering on the judge level this bias is especially severe, as per definition the cases with possible covariance due to this clustering dimension are handled by the same judge, and therefore all covariances are positively weighted.

Controlling for clustering at the judge level via the MD CJIVE or by including a FE for it is often not possible, however. The MD CJIVE would namely remove from $\bs P_Z$ all elements corresponding to cases handled by the same judge. In case there are no exogenous covariates these are the only non-zero elements of $\bs P_Z$, hence $\dddot{\bs P}_{Z}$ is a zero matrix making $\hat{\beta}_{MDCJ}$ undefined. Similarly, adding a FE for the judges and using the FE CJIVE is not possible, as the FE would be multicollinear with the instruments, such that $\hat{\beta}_{FECJ}$ is also undefined.

\section{Monte-Carlo}\label{sec:MC}
In this section we study the finite sample performance of the MD CJIVE and FE CJIVE in a Monte-Carlo simulation. We generate data from \cref{eq:judge IV} with the data clustered on two dimensions. We first group $n=500$ observations in $k=30$ groups and assign each group a judge. Next, we further group the observations in two times $G^{(1)}=G^{(2)}=30$ groups for the clustering in the two dimensions. The groups are constructed in the same way as in \citet{ligtenberg2023inference}. We first determine the group sizes and next randomly divide the observations over the clusters. In each clustering dimension $c=1,2$ we first set $n_g^{(c)}=\max\{1, n\exp(\gamma^{(c)} g/G^{(c)})/[\sum_{g=1}^{G^{(c)}-1}\exp(\gamma^{(c)} g/G^{(c)})+1]\}$ for $g=1,\dots G^{(c)}-1$ and $n^{(c)}_{G^{(c)}}=\max\{1,n-\sum_{g=1}^{G^{(c)}-1}n^{(c)}_g\}$. Next, all cluster sizes are ensured to be integer by rounding them down. Finally, to get the correct sample size, the first $n-\sum_{g=1}^{G^{(c)}}n^{(c)}_g$ clusters are increased by one. The unbalancedness of the clusters are determined by $\gamma^{(1)}=\gamma^{(2)}=2$, such that the smallest and largest clusters are 5 and 38. The groups for the judges are created in the same fashion with the same unbalancedness parameter.

Next, we generate the first stage errors $\eta_i=w_{1}\eta_{1,i}+w_{2}\eta_{2}+(1-w_{1}-w_{2})\eta_{3,i}$, where $\eta_{1,i}$ and $\eta_{2,i}$ are clustered components of the errors, one for each dimension, and $\eta_{3,i}$ is an idiosyncratic component. These components are weighted with $w_{1}=w_{2}=1/3$. We want to capture both the case in which FE can model all clustering and the more general case in which they cannot. Therefore we generate $\eta_{1,i}$ and $\eta_{2,i}$ from the factor model discussed, among others, in \citet{mackinnon2023cluster}. That is, for $c=1,2$ and $g=1,\dots,G^{(c)}$, $\bs\eta_{c,[g]^{(c)}}=(\sqrt{1-(\omega^{(c)})^2} u^{(c)}_{g}+\omega^{(c)}\bs e^{(c)}_{[g]^{(c)}})f_{g}$. In case $\omega^{(c)}=0$ all clustering can be captured by cluster FE. We draw $u^{(c)}\sim\N(0,1)$ and $f^{(c)}\sim\N(0,9)$. Since we showed in \cref{sec:FE for general} that the covariances of the errors must vary systematically with the elements of the projection matrix to generate a bias, we set $\bs e^{(c)}_{[g]^{(c)}}\sim\N(\bs 0,\bs\Sigma^{(c)}_g)$, with $\bs\Sigma^{(c)}_g=\bs D_{A_g^{(c)}}^{-1/2}\bs A_g^{(c)}\bs D_{A_g^{(c)}}^{-1/2}$, where $\bs A_g^{(c)}=(\bs P_{Z,[g,g]^{(c)}}+\bs I_{n^{(c)}_g}0.01)$. $\bs\Sigma^{(c)}_{g}$ is therefore a correlation matrix constructed from $\bs P_{Z,[g,g]^{(c)}}$ plus a constant that is added to ensure invertibility. The second stage errors are constructed from the first as $\varepsilon_{i}=\rho\eta_{i}+\sqrt{1-\rho^2}v_{i}$ for $\rho=0.5$ and $v_{i}\sim\N(0,1)$.

Finally, we set $\beta=0$ and $\bs\Pi$ equal to a draw from $\N(\bs 0,\bs I_{k})$, which we keep constant over the simulations.

\Cref{fig:b00} shows the boxplots for different estimators when estimated for $10\,000$ data sets from above DGP with $\omega^{(1)}=\omega^{(2)}=0$, such the clustering in both clustering dimensions can be fully captured by FE. In particular, we show 2SLS; JIVE; CJIVE that clusters only on the first clustering dimension; FE JIVE that includes FE for both clustering dimensions; FE CJIVE that models the first clustering dimension with FE and the second dimension with more general clustering; and the MD CJIVE that clusters on both dimensions. In case we were unable to calculate one of the estimators due to for example a singular matrix, we exclude that estimate from the analysis.

As expected, we observe a bias for the three estimators that do not model all clustering dimensions. For 2SLS the bias is the largest and it decreases for JIVE and CJIVE when more and more terms terms causing a bias are removed. For CJIVE the bias is even close to zero.

In this DGP all clustering can be captured by FE, we therefore see that the FE JIVE and FE CJIVE perform well, as they are unbiased. The MD CJIVE on the other hand seems to yield estimates that are on average and in median too low. This relatively poor behaviour of the MD CJIVE may be attributable to the small sample size.

From this figure it also becomes clear that in this DGP it is more efficient to model the clustering dimensions by FE than by leaving them completely free, as the FE JIVE is less dispersed than the FE CJIVE , which in turn is less dispersed than the MD CJIVE.

Next, we turn to \cref{fig:b01} which again shows the boxplot for the same estimators, but now for a DGP in which $\omega^{(1)}=0$ and $\omega^{(2)}=1$ such that only the clustering in the first dimension can be captured by FE. We see that again 2SLS, JIVE and CJIVE are biased. CJIVE is also more biased than in the first figure. More importantly, since the FE JIVE can no longer correctly model all clustering, we observe that also this estimator is biased. Also note that the bias of the MD CJIVE shrunk compared to the previous DGP.

Finally, in \cref{fig:b11} we show the same results for a DGP in which $\omega^{(1)}=\omega^{(2)}=1$ such that neither clustering dimension can be modelled by FE. In this case the bias of the MD CJIVE is close to zero, showing its robustness to general clustering. All other estimators are biased, however. Note in particular that the bias of the FE JIVE and FE CJIVE is similar in magnitude or bigger than the bias of the CJIVE estimator, which ignores one clustering dimension.

\begin{figure}[t]
    \centering
    \caption{Boxplot of different estimators for clustered data.}\label{fig:b00}
    \includegraphics{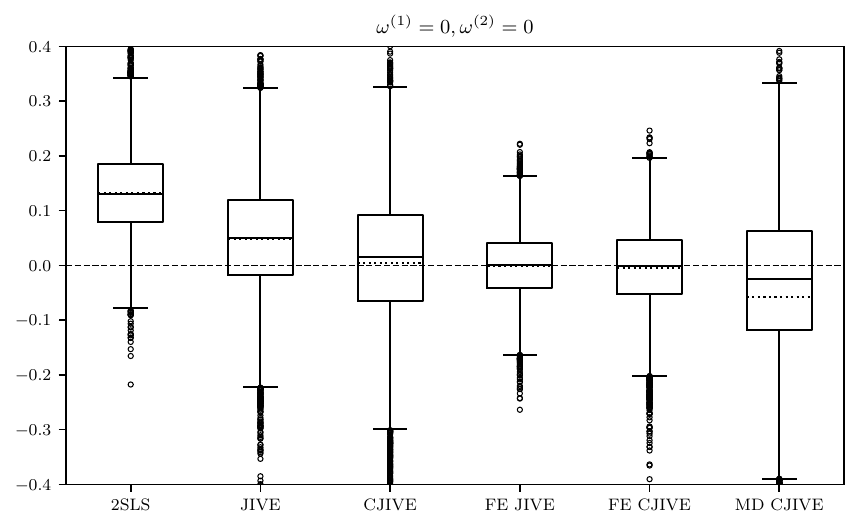}
    \begin{minipage}{0.95\textwidth}
    \footnotesize
    \textit{Note:} This figure gives a boxplot of 2SLS, JIVE, FE JIVE, FE CJIVE and MD CJIVE estimates for data clustered in two dimensions. The median is indicated by a solid line and the mean by a dotted line. $\omega^{(1)}$ and $\omega^{(2)}$ weigh in each clustering dimension the part of the errors that cannot be modelled by FE, such that both dimensions can be fully captured by FE. The DGP is described in \cref{sec:MC}.
    \end{minipage}
\end{figure}

\begin{figure}[t]
    \centering
    \caption{Boxplot of different estimators for clustered data.}\label{fig:b01}
    \includegraphics{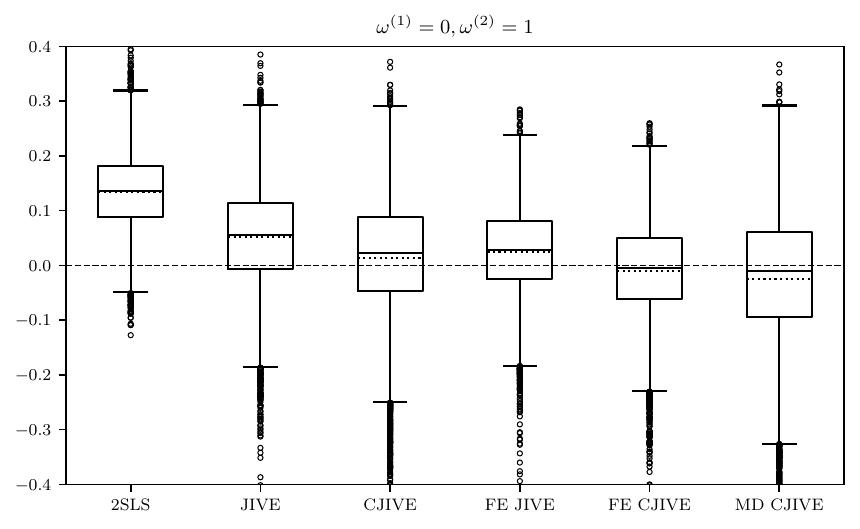}
    \begin{minipage}{0.95\textwidth}
    \footnotesize
    \textit{Note:} This figure gives a boxplot of 2SLS, JIVE, FE JIVE, FE CJIVE and MD CJIVE estimates for data clustered in two dimensions. The median is indicated by a solid line and the mean by a dotted line. $\omega^{(1)}$ and $\omega^{(2)}$ weigh in each clustering dimension the part of the errors that cannot be modelled by FE, such that only the first can be fully captured by FE. The FE CJIVE correctly models the second dimension with general clustering. The DGP is described in \cref{sec:MC}.
    \end{minipage}
\end{figure}

\begin{figure}[t]
    \centering
    \caption{Boxplot of different estimators for clustered data.}\label{fig:b11}
    \includegraphics{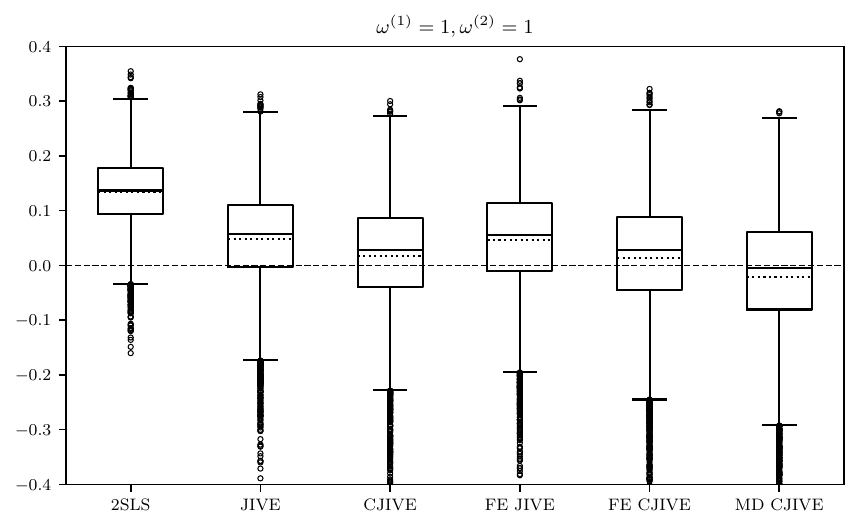}
    \begin{minipage}{0.95\textwidth}
    \footnotesize
    \textit{Note:} This figure gives a boxplot of 2SLS, JIVE, FE JIVE, FE CJIVE and MD CJIVE estimates for data clustered in two dimensions. The median is indicated by a solid line and the mean by a dotted line. $\omega^{(1)}$ and $\omega^{(2)}$ weigh in each clustering dimension the part of the errors that cannot be modelled by FE, such that neither clustering dimension can be fully captured by FE. The FE CJIVE models the second dimension with general clustering. The DGP is described in \cref{sec:MC}.
    \end{minipage}
\end{figure}

\section{Empirical illustrations}\label{sec:applications}
In this section we revisit two judge designs and show how the original results change when we account for clustering in multiple dimensions.

\subsection{\texorpdfstring{\citet{di2013criminal}}{Di Tella and Schargrodsky (2013)}}\label{ssec:DTS}
The first study that we re-investigate is the one by \citet{di2013criminal}. The authors estimate the effect of assigning a criminal to electronic monitoring (EM) as alternative to incarceration, on recidivism. For this they use that in the province of Buenos Aires, Argentina, alleged offenders are randomly assigned to judges. In addition, the judges' assignment rates of EM differ widely, which makes that a judge design can be used to estimate the effect of EM on recidivism.

To estimate this effect, \citet{di2013criminal} use a leave-out mean conviction rate, $\bs L$, as instrument. The model is augmented with exogenous covariates. That is
\begin{equation}\label{eq:model leniency exogenous}
    \begin{split}
        \bs y&=\bs X\beta+\bs W\bs\Gamma_2+\bs\varepsilon\\
        \bs X&=\bs L\pi+\bs W\bs\Gamma_1+\bs\eta.
    \end{split}
\end{equation}
Here, $\bs y$ are dummy variables indicating recidivism, $\bs X$ are dummy variables indicating whether EM was assigned, $\bs L$ is the leave-out average EM assignment rate per judge, where individuals are left out and $\bs W$ are exogenous covariates. The exogenous covariates, of which there are 41 in total, include indicators for the most serious crime committed divided over 11 categories, year and judicial district fixed effects.

The model is estimated using data on 1503 cases of people that were released from the penal system between the first of January 1998 and the twenty-third of October 2007, in the Province of Buenos Aires, Argentina. In 386 of these cases EM was assigned to the defendant. The sample is constructed such that each of the 171 judges assigns EM at least one time and has handled 10 cases or more in a larger data set.

By using the a leave-out EM assignment rate, \citet{di2013criminal} implicitly cluster at the individual level and remove the many instrument bias stemming from this clustering dimension. The reported standard errors are furthermore clustered at judicial districts, but this is not taken into account for the coefficient estimates. Additional clustering dimensions may emerge when the FE are projected out.

In \cref{tab:DTS} we re-estimated $\beta$ accounting for the different clustering dimensions using the MD CJIVE and FE CJIVE, for which we adopt the model as in \cref{eq:judge IV exogenous}. For both estimators we started with the model that naively projects out the controls and clusters on the individual level as defined by the leave-out mean instrument in the original paper. Next, we keep adding clustering dimensions implied either by the standard errors or by the FE. The MD CJIVE takes care of the clustering by removing the terms from the estimator. The FE CJIVE adds FE for the clustering dimensions implied by the standard errors and correctly removes these. When the clustering dimensions implied by the FE are added to the FE CJIVE we also correctly remove these FE. Finally, for this estimator we give the estimates obtained by removing all control variables via the FE CJIVE.

We start with the estimates obtained taking the clustering into account via the MD CJIVE as given in the second column of \cref{tab:DTS}. By leaving out the terms corresponding to the judicial district and the most serious crime, the estimates are almost one and a half times larger in magnitude. If we also model the last clustering dimension implied by the FE by leave-out the estimate jumps back close to the original level, but stays slightly more negative.

The estimates from the FE CJIVE, as given in the third column of \cref{tab:DTS}, are unstable. The first estimate, which only clusters at the individual level, is the same as the one by the MD CJIVE, since it does not remove any control variables yet. Throughout we use the individual level as the level that is modelled via leave-out. Adding FE for the judicial district and correctly removing these makes the estimates more negative. Perhaps even unrealistically so, since it would imply that criminals on the margin have more than 60\% lower recidivism rates if they receive EM instead of a prison sentence.

Also correctly removing the FE for most serious crime makes the estimates even more negative. If we remove the year FE with the FE CJIVE on the other hand the estimates turn positive. Finally, in the last line we give the estimate of the FE CJIVE with all control variables. This estimate is negative and large in magnitude, but does not have a sensible interpretation, however, as it is far outside the zero and one range.

\begin{table}[ht]
    \centering
    \caption{Effect of electronic monitoring on recidivism.}
    \label{tab:DTS}
    \begin{tabular}{lrr}
\toprule
 Clusters                 &   MD CJIVE &   FE CJIVE \\
\midrule
 Individual               &    -0.0997 &    -0.0997 \\
 + Judicial district (SE) &    -0.1059 &    -0.6155 \\
 + Crime (FE)             &    -0.1434 &    -0.7748 \\
 + Year (FE)              &    -0.1057 &     0.5529 \\
 + All controls           &            &    -2.9520 \\
\bottomrule
\end{tabular}
\\[1ex]
\begin{minipage}{0.95\textwidth}
    \footnotesize
    \textit{Note:} This table shows estimates effect of electronic monitoring on recidivism using the data from \citet{di2013criminal}. The estimators account for different levels of clustering as indicated by the first column, where we specify whether the clustering is coming from how the standard errors are specified (SE) or from removing the fixed effects (FE). The FE CJIVE models the first clustering dimension by leave-out and the others by correctly removing the FE. The final row gives the estimate obtained by removing all control variables via the FE CJIVE.
\end{minipage}
\end{table}

\subsection{\texorpdfstring{\citet{agan2023misdemeanor}}{Agan et al. (2023)}}
The second study we revisit is the one by \citet{agan2023misdemeanor}. This study estimates the effect of \emph{not} prosecuting someone who committed a nonviolent misdemeanour offence on subsequent criminal behaviour. The authors use that in Suffolk County, Massachusetts, cases of nonviolent misdemeanours are randomly assigned to judges, conditional on the time, the day of the week and the location. The varying strictness of the judges can therefore be used to identify the effect of nonprosecution.

The model to estimate this effect is similar to \cref{eq:model leniency exogenous} from the previous section. Now $\bs y$ indicates whether the defendant faced a new criminal charge within two years after having committed the nonviolent misdemeanour, $\bs X$ is an indicator variable for non-prosecution and $\bs W$ capture individual characteristics and court and time FE. The instrument is constructed differently from before, however. Since the types of misdemeanours vary over time, the day of the week and the location, \citet{agan2023misdemeanor} residualise the prosecution decisions with respect to FE for these variables before taking a leave-out average. $\bs L$ is therefore a residualised leave-out average prosecution rate per judge, where individuals are left out.

The model is estimated using data on misdemeanour criminal complaints in Suffolk County, Massachusetts, between the first of January, 2004, and the first of January, 2020. In total this data set contains 67\,060 cases involving 53\,358 individuals. The cases are handled by 315 judges. In addition to the prosecution decision, whether the defendant faced new charges within two years and the ADA identity, the data set also contains characteristics of the defendant. \citet{agan2023misdemeanor} use 16 of these as controls. Together with an intercept and the court-day-of-week and court-month fixed effects, this brings the dimension of $\bs W_{i}$, after removing redundant controls, to 1276.  

In their main analysis \citet{agan2023misdemeanor} cluster the standard errors at the individual and judge level. The many instrument bias due to clustering on the individual level is correctly removed by using the leave-out average prosecution rate. The many instrument bias from clustering on the judge level, on the other hand, is not addressed and can also not be taken care of using the methods developed in this paper. Other relevant clustering may be introduced by projecting out the court-time FE.

In \cref{tab:ADH} we show again estimates accounting for different clustering dimensions similar to the previous subsection. Note that, we adopt the model as in \cref{eq:judge IV exogenous} and hence in the baseline model we project out the exogenous control variables differently than in the original study. Contrary to the results in the previous subsection, we see that the MD CJIVE estimates become less negative when taking more clustering dimensions into consideration, which implies a weaker effect of nonprosecution on recidivism. In particular the final estimate is much smaller magnitude than the previous two. This may be due to the court-month clusters being relatively large, which means that many terms need to be removed from the estimator, which can change its value also relatively much.

The FE CJIVE shows again less stable behaviour than the MD CJIVE. The first two and final estimates are all close to each other, but the third estimate deviates. Also, as opposed to the estimates from the MD CJIVE, we observe accounting for more clustering makes the FE CJIVE more negative.

\begin{table}[ht]
    \centering
    \caption{Effect of nonprosecution on recidivism.}
    \label{tab:ADH}
    \begin{tabular}{lrr}
\toprule
 Clusters                 &   MD CJIVE &   FE CJIVE \\
\midrule
 Individual               &    -0.1996 &    -0.1996 \\
 + Court-day of week (FE) &    -0.1805 &    -0.1993 \\
 + Court-month (FE)       &    -0.1274 &     0.0242 \\
 + All controls           &            &    -0.2004 \\
\bottomrule
\end{tabular}
\\[1ex]
\begin{minipage}{0.95\textwidth}
    \footnotesize
    \textit{Note:} This table shows estimates effect of nonprosecution on recidivism using the data from \citet{agan2023misdemeanor}. The estimators account for different levels of clustering as indicated by the first column, where we indicate when the clustering stems from removing the fixed effects (FE). The FE CJIVE models the first clustering dimension by leave-out and the others by correctly removing the FE. The final row gives the estimate obtained by removing all control variables via the FE CJIVE.
\end{minipage}
\end{table}

\section{Concluding remarks}\label{sec:conclusion}
In this paper we argued that most judge designs imply that their data are clustered in multiple dimensions due to how they specify their estimator, their standard errors and control variables. Nevertheless, only a single clustering dimension is taken into account by the 2SLS estimator with a leave-out instrument. Consequently and as we show in Monte-Carlo simulations, a many instrument bias may persist.

Conventional estimators cannot handle multidimensional clustered data. We therefore proposed the MD CJIVE and FE CJIVE that either remove from the conventional estimators the terms associated with the bias or eliminate the clustering by including fixed effects. In Monte-Carlo simulations we found that these estimators are unbiased as long as the modelling assumption underlying them are satisfied. Furthermore, we show that modelling the clustering with FE might work well even if the real data clustering is more complex. However, these two new estimators are in general not a solution to a bias that may emerge due to clustering on the judge level, which is an especially popular clustering level. 

Next, we showed the empirical relevance of clustering in judge designs by revisiting the studies from \citet{di2013criminal} and \citet{agan2023misdemeanor}. We found that taking the clustering into account via the MD CJIVE usually yields more stable results than via the FE CJIVE and that these estimates can differ from those that ignore the clustering.

Finally, for future research we note that the adaptations we made to the projection matrix in the 2SLS estimator to obtain the MD CJIVE and FE CJIVE might also work in other methods for IV models. For example, our techniques can potentially generalise the HLIM and HFUL estimators \citep{hausman2012instrumental} or the many instrument and identification robust tests \citep{crudu2021inference,mikusheva2021inference,ligtenberg2023inference,matsushita2020jackknife} to multidimensional clustered data.

\bibliographystyle{chicago}
\bibliography{Bib}

\newpage
\begin{appendices}
\crefalias{section}{appsec}
\crefalias{subsection}{appsubsec}
\numberwithin{equation}{section}
\numberwithin{counter}{section}

\section{Overview of judge designs}\label{app:overview}
In this section we review judge designs to show the relevance of data clustering. To prevent selection bias, we started with the 70 papers listed in the overview Table A1 in \citet{chyn2024examiner}. We were unable to find two papers on the list, two other papers were merged into one and we excluded a final two papers because they used simulated data or were not a judge design, which brings the total number of papers to 65.

The papers are listed in \cref{tab:overview} and for each of them we document a couple of characteristics. First of all, we note which estimator each paper used. It strikes that the vast majority, 55 to be precise, used 2SLS with a leave-out instrument or a variant thereof. Another three papers used the judge dummies as instrument and estimated the effect of a conviction with JIVE.

A leave-out instrument or JIVE implicitly allows for clustering in a single dimension. The second characteristic we document is therefore this clustering dimension. For exposition purposes we discern four clustering categories labelled case, individual, judge and other. Case refers to a single observation, individual to one entity that can be involved in multiple cases, judge to the decision maker and other to all other possible dimensions.

Next, we document the dimensions on which the standard errors were clustered in the main specification and in the robustness checks. In case the papers used two-way clustered standard errors, we list both dimensions. From the table we see that 49 papers that used a form of leave-out or jackknife controlled in their main specification for clustering in another dimension than that they left out in estimation. An additional two papers allowed for clustering in other dimensions than what is left out in the robustness checks. This shows that researchers are concerned with dependence in multiple dimensions, but they do not properly control for this in estimation.

That clustering in multiple dimensions is empirically relevant, also becomes clear by the eleven papers that clustered their standard errors in two dimensions in the main specification. Moreover, there are an additional seven papers that clustered their standard errors in a single dimension, but in which the authors suspected additional dependence, hence they clustered on additional dimensions in robustness checks.

From the overview it also becomes clear that many papers clustered on the level of the judge. 30 papers did so in the main specification and eight in the robustness checks.

Finally, we document for each paper whether they included FE as control variables. Without exception they all do, which can induce additional dimensions of dependence depending on the number of FE relative to the sample size.

\begin{landscape}
    \begin{ThreePartTable}
\begin{TableNotes}
\footnotesize
\item[a]\label{tn:lo variant} This paper uses a variant of leave-out.
\item[b]\label{tn:ao lo} Amongst others the same clustering dimension as what is left out in the estimator.
\item[c]\label{tn:cl lo} Same clustering dimension as what is left out in the estimator.
\item[d]\label{tn:md cl} Includes clustering in multiple dimensions.
\note This table gives an overview of the estimators and clustering dimensions in the judge designs listed in \citet{chyn2024examiner}. In the column ``Estimator'' we list the estimator used, where JIVE refers to JIVE on judge dummies; 2SLS refers to 2SLS on judge dummies; leave-out refers to 2SLS with a leave-out instrument; SP jackknife refers to a semi-parametric estimation method that uses jackknife; mean refers to 2SLS with a mean conviction rate; and TS 2SLS refers to two sample 2SLS. The column ``Left-out'' specifies what is left out when JIVE or a leave-out instrument is used. The columns under ``Standard error cluster'' list the first and second clustering dimensions of the standard errors and the clustering dimensions of the standard errors in the robustness checks. For the clustering dimensions we use the terminology of a judicial judge design for classification. The column ``FE'' specifies whether the model included fixed effects. \\
\end{TableNotes}

\begin{longtable}{lllllll}
\caption{Overview of papers using a judge design.}\label{tab:overview}\\
\toprule
Paper & Estimator & Left-out &\multicolumn{3}{l}{Standard error cluster} & FE \\
\cmidrule{4-6}
&&& First & Second & Robustness\\
\midrule\endfirsthead
\toprule
Paper & Estimator & Left-out &\multicolumn{3}{l}{Standard error cluster} & FE \\
\cmidrule{4-6}
&&& First & Second & Robustness\\
\midrule\endhead
\bottomrule
\multicolumn{6}{l}{Table continues on the next page.}\\
\endfoot
\bottomrule
\insertTableNotes
\endlastfoot
\citet{kling2006incarceration}	&	JIVE	&	Case	&	Other	&		&		&	\checkmark\\
\citet{green2010using}	&	2SLS	&		&	Other	&		&		&	\checkmark\\
\citet{loeffler2013does}	&	2SLS	&		&		&		&		&	\checkmark\\
\citet{di2013criminal}	&	Leave-out	&	Individual	&	Other	&		&	Judge	&	\checkmark\\
\citet{aizer2015juvenile}	&	Leave-out	&	Case	&	Other	&		&	Judge	&	\checkmark\\
\citet{mueller2015criminal}	&	SP jackknife	&	Individual	&	Individual	&		&	Other	&	\checkmark\\
\citet{gupta2016heavy}	&	Leave-out	&	Individual	&	Other	&		&	Individual	&	\checkmark\\
\citet{leslie2017unintended}	&	Leave-out	&	Individual	&	Individual	&	Other	&	Judge, other	&	\checkmark\\
\citet{bhuller2018intergenerational}	&	Leave-out	&	Case	&	Judge	&	Individual	&		&	\checkmark\\
\citet{arteaga2018cost}	&	Leave-out	&	Individual	&	Judge	&	Other	&		&	\checkmark\\
\citet{arnold2018racial}	&	Leave-out	&	Individual	&	Individual	&	Other	&	Ind., judge\tnotex{tn:md cl}	&	\checkmark\\
\citet{dobbie2018effects}	&	Leave-out	&	Individual	&	Individual	&	Judge	&		&	\checkmark\\
\citet{stevenson2018distortion}	&	JIVE	&	Case	&		&		&		&	\checkmark\\
\citet{ribeiro2019pretrial}	&	Leave-out	&	Individual	&	Individual	&	Judge	&		&	\checkmark\\
\citet{white2019misdemeanor}	&	Mean	&		&		&		&		&	\checkmark\\
\citet{bhuller2020incarceration}	&	Leave-out	&	Case	&	Individual	&		&		&	\checkmark\\
\citet{zapryanova2020effects}	&	Leave-out	&	Individual	&		&		&		&	\checkmark\\
\citet{didwania2020immediate}	&	Leave-out	&	Individual	&	Other	&		&		&	\checkmark\\
\citet{norris2021effects}	&	Leave-out	&	Case	&	Individual	&	Other	&	Judge	&	\checkmark\\
\citet{agan2023misdemeanor}	&	Leave-out	&	Individual	&	Individual	&	Judge	&		&	\checkmark\\
\citet{arbour2021can}	&	Leave-out	&	Individual	&	Other	&		&	Other	&	\checkmark\\
\citet{bhuller2021prison}	&	Leave-out	&	Individual	&	Judge	&	Individual	&		&	\checkmark\\
\citet{eren2021juvenile}	&	Leave-out	&	Individual	&	Judge	&		&	Other	&	\checkmark\\
\citet{grau2023effect}	&	Leave-out	&	Individual	&	Judge	&		&		&	\checkmark\\
\citet{augustine2022impact} 	&	Leave-out	&	Case	&	Judge	&		&		&	\checkmark\\
\citet{alexeev2022fines}	&	Leave-out	&	Individual	&	Individual	&	Judge	&	Other	&	\checkmark\\
\citet{gonccalves2023police}	&	Leave-out	&	Case	&	Judge	&		&		&	\checkmark\\
\citet{humphries2023conviction}	&	Leave-out	&	Case	&	Judge	&		&		&	\checkmark\\
\citet{chang2006effect}	&	2SLS	&		&		&		&		&	\checkmark\\
\citet{dobbie2015debt}	&	Leave-out	&	Individual	&	Other	&		&		&	\checkmark\\
\citet{dobbie2017consumer}	&	Leave-out	&	Individual	&	Other	&		&		&	\checkmark\\
\citet{rieber2019herding}	&	JIVE	&	Individual	&		&		&	Individual	&	\checkmark\\
\citet{dippel2020property}	&	2SLS	&		&		&		&		&	\checkmark\\
\citet{honigsberg2021deleting}	&	Leave-out\tnotex{tn:lo variant}	&	Case	&	Other	&		&		&	\checkmark\\
\citet{cheng2021consumers}	&	Leave-out	&	Individual	&	Judge	&		&		&	\checkmark\\
\citet{cespedes2021effect}	&	Leave-out	&	Individual	&	Other	&		&		&	\checkmark\\
\citet{doyle2017evaluating}	&	Leave-out	&	Individual	&	Other	&		&	Judge	&	\checkmark\\
\citet{bakx2020better}	&	Leave-out	&	Individual	&		&		&		&	\checkmark\\
\citet{blaehr2021instrumental}	&	Leave-out	&	Individual	&		&		&		&	\checkmark\\
\citet{hertzberg2021overdiagnosing}	&	Leave-out	&	Individual	&	Judge	&		&		&	\checkmark\\
\citet{mullainathan2022diagnosing}	&	Leave-out	&	Individual	&		&		&		&	\checkmark\\
\citet{chan2022selection}	&	Leave-out	&	Individual	&	Judge	&		&		&	\checkmark\\
\citet{eichmeyer2022pathways}	&	Leave-out	&	Individual	&	Judge	&		&		&	\checkmark\\
\citet{galasso2015patents}	&	Leave-out\tnotex{tn:lo variant}	&	Case	&		&		&	Ind., other	&	\checkmark\\
\citet{galasso2018patent} 	&	Leave-out\tnotex{tn:lo variant}	&	Case	&	Individual	&		&		&	\checkmark\\
\citet{sampat2019patents}	&	TS 2SLS	&		&	Other	&		&	Individual	&	\checkmark\\
\citet{farre2020patent}	&	Leave-out\tnotex{tn:lo variant}	&	Case	&	Judge	&		&		&	\checkmark\\
\citet{doyle2007child}	&	Leave-out	&	Other	&	Judge	&		&		&	\checkmark\\
\citet{doyle2008child}	&	Leave-out	&	Other	&	Judge	&		&		&	\checkmark\\
\citet{maestas2013does}	&	Leave-out	&	Individual	&	Judge	&		&		&	\checkmark\\
\citet{doyle2013causal}	&	Leave-out	&	Other	&	Judge	&		&	Other\tnotex{tn:ao lo}	&	\checkmark\\
\citet{dahl2014family}	&	Leave-out	&	Case	&	Judge	&		&		&	\checkmark\\
\citet{french2014effect}	&	Leave-out	&	Case	&	Judge	&		&		&	\checkmark\\
\citet{hyman2018can}	&	Leave-out	&	Individual	&	Judge	&		&		&	\checkmark\\
\citet{collinson2024eviction}	&	Leave-out	&	Case	&	Other	&		&		&	\checkmark\\
\citet{hjalmarsson2019causal}	&	Leave-out	&	Individual	&	Judge	&		&	Other	&	\checkmark\\
\citet{autor2019disability}	&	Leave-out	&	Case	&	Judge	&		&		&	\checkmark\\
\citet{diamond2020effect}	&	Leave-out	&	Other	&	Other\tnotex{tn:cl lo}	&		&		&	\checkmark\\
\citet{black2024effect}	&	Leave-out	&	Case	&	Judge	&		&		&	\checkmark\\
\citet{gross2022temporary}	&	Leave-out	&	Case	&	Individual	&		&	Judge, ind., other\tnotex{tn:md cl}	&	\checkmark\\
\citet{bald2022causal}	&	Leave-out	&	Other	&	Judge	&	Other\tnotex{tn:cl lo}	&		&	\checkmark\\
\citet{baron2022there}	&	Leave-out	&	Case	&	Individual	&		&	Judge, ind., other\tnotex{tn:md cl}	&	\checkmark\\
\citet{silver2022impacts}	&	Leave-out	&	Individual	&	Other	&		&		&	\checkmark\\
\pagebreak
\citet{cohen2024housing}	&	Leave-out	&	Case	&	Judge	&		&		&	\checkmark\\
\citet{lee2023halfway}	&	Leave-out	&	Case	&	Judge	&		&		&	\checkmark\\
\end{longtable}
\end{ThreePartTable}
\end{landscape}

\section{Equivalence 2SLS and CJIVE}\label{app:equivalence}
To show equivalence between the 2SLS with a leave-out instrument and the CJIVE with the judge identity dummies, we consider the model from \cref{eq:judge IV} with clustering on a single dimension. Let $L_i$ be the number of convictions of the judge $J(i)$  excluding all cases involving the defendant of case $i$. Also let $\tilde{\bs D}_n$ be the $n\times n$ diagonal matrix, with $i\th$ diagonal element the number of cases that the judge of case $i$ handles, excluding any cases that involve the defendant of case $i$. The leave-out mean leniency measure can then be written as $\tilde{\bs D}_n^{-1}\bs L$ and the 2SLS estimator using this instrument is $\hat{\beta}_{2SLS}=(\bs X'\tilde{\bs D}_n\bs L)^{-1}\bs L'\tilde{\bs D}_{n}\bs y$.

To see that this is similar to a cluster JIVE with clustering on the individual, note that $[\bs Z'\bs Z]_{ij}=\sum_{l=1}^nZ_{li}Z_{lj}$. Each term can only be nonzero when $i=j$ because there is only one judge per case. $Z_{li}$ equals one for every case $j$ handles. Hence, $\bs Z'\bs Z$ is a $k\times k$ diagonal matrix with diagonal element $i$ equal to $n_{i}$, the number of cases assigned to judge number $i$. To emphasise that it is a diagonal matrix we write it as $\bs D_k$. Similarly, we define $\bs D_n$ as the $n\times n$ diagonal matrix with diagonal element $i$ equal to $n_{J(i)}$, the number of cases that the judge assigned to case $i$ handles.

Also, $[\bs Z\bs Z']_{ij}=\sum_{l=1}^kZ_{il}Z_{jl}$, which equals one when the defendants in cases $i$ and $j$ are handled by the same judge. Therefore $\bs Z\bs Z'$ is a matrix with in row $i$, ones in the columns corresponding to the cases also handled by the judge assigned to case $i$ and zeroes elsewhere. Combining the two we find, $\bs P_{Z}=\bs Z(\bs Z'\bs Z)^{-1}\bs Z'=\bs Z\bs D_k^{-1}\bs Z'=\bs D_n^{-1}\bs Z\bs Z'$ with 
\begin{equation}\label{eq:Pz}
    \begin{split}
        P_{Z,ij}&=\begin{cases}
            1/n_{J(i)} &\text{ if } J(i)=j\\
            0 &\text{ otherwise}.
        \end{cases}
    \end{split}
\end{equation}
Hence, $\bs P_{Z}\bs X$ has as $i\th$ element the average number of convictions of $J(i)$.

Remember that we cluster on the individual level, such that $\bs B_{ZZ'}$ equals a matrix with in row $i$ ones in the columns corresponding to the cases involving the defendant of case $i$ that were assigned to $J(i)$ and zeroes elsewhere. Consequently, the $i\th$ element of $\bs B_{Z Z'}\bs X$ is the number of times $J(i)$ convicted defendant $i$. This gives $\bs B_{P_Z}\bs X=\bs B_{Z(Z'Z)^{-1}Z'}\bs X=\bs B_{D_n^{-1}ZZ'}\bs X=\bs D_n^{-1}\bs B_{ZZ'}\bs X$, with element $i$ the number of times $J(i)$ convicted the defendant $i$ divided by $n_{J(i)}$.

Therefore $\bs Z\bs Z'\bs X-\bs B_{ZZ'}\bs X=\bs L$ is the leave-out number of convictions. Next, note that  $\ddot{\bs P}_{Z}\bs X=\bs P_{Z}\bs X-\bs B_{P_{Z}}\bs X=\bs D_n^{-1}[\bs Z\bs Z'\bs X-\bs B_{ZZ'}\bs X]=\bs D_n^{-1}\bs L$. Such that we can write the cluster JIVE estimator as $\hat{\beta}_{CJIVE}=(\bs X'\ddot{\bs P}_{Z}\bs X)^{-1}\bs X'\ddot{\bs P}_{Z}\bs y=(\bs X'\bs D_n^{-1}[\bs Z\bs Z'\bs X-\bs B_{ZZ'}\bs X])^{-1}[\bs X'\bs Z\bs Z'-\bs X'\bs B_{ZZ'}]\bs D_n^{-1}\bs y=(\bs X'\bs D_n^{-1}\bs L)^{-1}\bs L'\bs D_n^{-1}\bs y$. This can be interpreted as a weighted version of the 2SLS estimator using the leave-out instrument.

\section{Consistency}
\subsection{MD CJIVE}\label{app:consistency MD CJIVE}
In this section we use the following additional notation. For $\bs A$ a square matrix let $\lambdamin(\bs A)$ and $\lambdamax(\bs A)$ be the minimum and maximum eigenvalues of $\bs A$. For $\bs A$ not necessarily square $\|\bs A\|_2=\sqrt{\lambdamax(\bs A'\bs A)}$ is the spectral norm. $\toprob$ denotes convergence in probability. $\bar{C}$ denotes a finite positive scalar that independent from the sample size that can differ between appearances.

Consider the model as in \cref{eq:judge IV}. We obtain consistency of the MD CJIVE estimator via high level assumptions similar to those in \citet{frandsen2023cluster}.

\begin{assumption}\label{ass:consistency MD CJIVE}
Conditional on $\bs Z$, $\bs X'\dddot{\bs P}_{Z}\bs X/n\toprob\bs Q_{XPX}$ a positive definite matrix and $\bs X'\dddot{\bs P}_{Z}\bs\varepsilon/n\toprob\lim_{n\to\infty}\E(\bs X'\dddot{\bs P}_{Z}\bs\varepsilon/n|\bs Z)$.
\end{assumption}

This assumption ensures that the limit of the inverse in the MD CJIVE exists and that a law of large numbers works on the part after the inverse.

\begin{theorem}
Under \cref{ass:consistency MD CJIVE} it holds that $\hat{\bs\beta}_{MDCJ}-\bs\beta\toprob0$.
\end{theorem}

\begin{proof}
By Slutsky's lemma and $\E(\bs X'\dddot{\bs P}_{Z}\bs\varepsilon|\bs Z)=\E(\bs\Pi'\bs Z'\dddot{\bs P}_{Z}\bs\varepsilon+\bs\eta'\dddot{\bs P}_{Z}\bs\varepsilon|\bs Z)=0$, we have that $\hat{\bs\beta}_{MDCJ}-\bs\beta\toprob\bs Q_{XPX}^{-1}(\lim_{n\to\infty}\E(\bs X'\dddot{\bs P}_{Z}\bs\varepsilon/n|\bs Z))=0$.
\end{proof}

\subsection{FE CJIVE}\label{app:consistency FE CJIVE}
In this section we consider a slightly more general version of \cref{eq:judge IV exogenous}, where there are $p$ endogenous regressors. That is
\begin{equation}
    \begin{split}
        \bs y&=\bs X\bs\beta+\bs W\bs\Gamma_2+\bs\varepsilon\\
        \bs X&=\bs Z\bs\Pi+\bs W\bs\Gamma_1+\bs\eta,
    \end{split}
\end{equation}
where now $\bs X\in\R^{n\times p}$. For ease of notation we again let $\tilde{\bs P}_{Z}=\bs P_{M_{W}Z}-\bs M_{Z,W}\bs H\bs M_{Z,W}$ the FE CJIVE projection matrix.

We show consistency of the FE CJIVE under many weak instrument asymptotics as in \citet{chao2023jackknife}. Make the following assumptions.
\begin{assumption}\label{ass:rate}
$\bs\Pi=\bs\Upsilon\bs D_{\mu}/\sqrt{n}$, with $\mu_{\min}=\min_{j=1,\dots,p}\mu_{j}$ such that for all $\mu_{j}$ either $\mu_{j}=\sqrt{n}$ or $\mu_{j}/\sqrt{n}\to 0$. Also, $\mu_{\min}\to\infty$ as $n\to\infty$ such that $\sqrt{k}\nmax/\mumin^2\to0$ and $\nmax/\mumin\to0$ for $\nmax=\max_{g=1,\dots,G}n_g$.
\end{assumption}

\begin{assumption}\label{ass:moments}
Conditional on $\bs Z$ and with probability one for all $n$ large enough and for any $\bs v\in\R^p$ such that $\bs v'\bs v=1$ and 
    \begin{equation}
        \begin{split}
            \bs M_g(\bs v)=\E(\begin{bmatrix}
                \bs\varepsilon_{[g]}\bs\varepsilon_{[g]}' & \bs\varepsilon_{[g]}\bs v'\bs\eta_{[g]}' \\
                \bs\eta_{[g]}\bs v\bs\varepsilon_{[g]}' & \bs\eta_{[g]}\bs v\bs v'\bs\eta_{[g]}' 
            \end{bmatrix}|\bs Z)=\begin{bmatrix}
                \bs\Sigma_{g} &\bs\Xi(\bs v)_{g}'\\
                \bs\Xi(\bs v)_{g} &\bs\Omega(\bs v)_{g}
            \end{bmatrix},
        \end{split}
    \end{equation}
    it holds that $0<1/\bar{C}\leq \lambdamin(\bs M_g(\bs v))\leq\lambdamax(\bs M_g(\bs v))\leq n_g\bar{C}<\infty$.
\end{assumption}

\begin{assumption}\label{ass:asymptotic}
Conditional on $\bs Z$ and with probability one for $n$ large enough it holds that \begin{assenumerate*}
    \item\label{assit:ZZ} $\lambdamax(\bs\Upsilon'\bs Z'\bs Z\bs\Upsilon/n)\leq\bar{C}$; 
    \item\label{assit:Pmwz} $\lambdamin(\bs\Upsilon'\bs Z'\bs P_{M_{W}Z}\bs Z\bs\Upsilon/n)\geq 1/\bar{C}>0$;
    \item\label{assit:trHH} $\tr(\bs H\bs H)\leq \bar{C}k$.
\end{assenumerate*} 
\end{assumption}

\begin{lemma}\label{lem:XPX}
Under \cref{ass:rate,ass:moments,ass:asymptotic} $\bs D_{\mu}^{-1}\bs X'\tilde{\bs P}\bs X\bs D_{\mu}=O_p(1)$.
\end{lemma}

\begin{proof}
Since $\tilde{\bs P}\bs W=\bs 0$ we have
\begin{equation}
    \begin{split}
        &\bs D_{\mu}^{-1}\bs X'\tilde{\bs P}\bs X\bs D_{\mu}^{-1}\\
        &=\bs D_{\mu}^{-1}(\bs Z\bs\Pi+\bs W\bs\Gamma_1+\bs\eta)'\tilde{\bs P}(\bs Z\bs\Pi+\bs W\bs\Gamma_1+\bs\eta)\bs D_{\mu}^{-1}\\
        &=\bs D_{\mu}^{-1}\bs\Pi\bs Z'\bs P\bs Z\bs\Pi\bs D_{\mu}^{-1}+\bs D_{\mu}^{-1}\bs\Pi\bs Z'\tilde{\bs P}\bs\eta\bs D_{\mu}^{-1}+\bs D_{\mu}^{-1}\bs\eta'\tilde{\bs P}\bs Z\bs\Pi\bs D_{\mu}^{-1}+\bs D_{\mu}^{-1}\bs\eta'\tilde{\bs P}\bs\eta\bs D_{\mu}^{-1}\\
        &=\bs\Upsilon'\bs Z'\tilde{\bs P}\bs Z\bs\Upsilon/n+\bs\Upsilon'\bs Z'\tilde{\bs P}\bs\eta\bs D_{\mu}^{-1}/\sqrt{n}+\bs D_{\mu}^{-1}\bs\eta'\tilde{\bs P}\bs Z\bs\Upsilon/\sqrt{n}+\bs D_{\mu}^{-1}\bs\eta'\tilde{\bs P}\bs\eta\bs D_{\mu}^{-1}.
    \end{split}
\end{equation}

For the first term, let $\bs a,\bs b\in\R^{p}$ such that $\|\bs a\|=\|\bs b\|=1$ be given. Then by the Cauchy-Schwartz inequality
\begin{equation}
    \begin{split}
        \bs a'\bs\Upsilon'\bs Z'\tilde{\bs P}\bs Z\bs\Upsilon\bs b/n&=\bs a'\bs\Upsilon'\bs Z'(\bs P_{M_{W}Z}-\bs M_{Z,W}\bs H\bs M_{Z,W})\bs Z\bs\Upsilon\bs b/n\\
        &=\bs a'\bs\Upsilon'\bs Z'\bs P_{M_{W}Z}\bs Z\bs\Upsilon\bs b/n\\
        &=\bs a'\bs\Upsilon'\bs Z'\bs M_{W}\bs Z\bs\Upsilon\bs b/n\\
        &\leq[\bs a'\bs\Upsilon'\bs Z'\bs M_{W}\bs Z\bs\Upsilon\bs a/n\bs b'\bs\Upsilon'\bs Z'\bs M_{W}\bs Z\bs\Upsilon\bs b/n]^{1/2}\\
        &\leq[\bs a'\bs\Upsilon'\bs Z'\bs Z\bs\Upsilon\bs a/n\bs b'\bs\Upsilon'\bs Z'\bs Z\bs\Upsilon\bs b/n]^{1/2}.
    \end{split}
\end{equation}
Therefore, under \cref{assit:ZZ} $\bs\Upsilon'\bs Z'\tilde{\bs P}\bs Z\bs\Upsilon/n=O_p(1)$.

For the second term and third term take $\bs a,\bs b$ as before. Then write
\begin{equation}\label{eq:ZPPZ}
    \begin{split}
        \E([\bs a'\bs\Upsilon'\bs Z'\tilde{\bs P}\bs\eta\bs D_{\mu}^{-1}\bs b/\sqrt{n}]^2|\bs Z)&=\frac{1}{n}\E(\bs a'\bs\Upsilon'\bs Z'\tilde{\bs P}\bs\eta\bs D_{\mu}^{-1}\bs b\bs b'\bs D_{\mu}^{-1}\bs\eta'\tilde{\bs P}\bs Z\bs\Upsilon\bs a|\bs Z)\\
        &=\frac{1}{n}\bs a'\bs\Upsilon'\bs Z'\tilde{\bs P}\E(\bs\eta\bs D_{\mu}^{-1}\bs b\bs b'\bs D_{\mu}^{-1}\bs\eta'|\bs Z)\tilde{\bs P}\bs Z\bs\Upsilon\bs a\\
        &\leq\frac{\bar{ C}\nmax}{n\mumin^2}\bs a'\bs\Upsilon'\bs Z'\tilde{\bs P}\tilde{\bs P}\bs Z\bs\Upsilon\bs a
        =O_p(\frac{n_{\max}}{\mu_{\min}^2}),
    \end{split}
\end{equation}
by \cref{ass:moments} and since
\begin{equation}
    \begin{split}
        \frac{1}{n}\bs a'\bs\Upsilon'\bs Z'\tilde{\bs P}\tilde{\bs P}\bs Z\bs\Upsilon\bs a
        &=\frac{1}{n}\bs a'\bs\Upsilon'\bs Z'(\bs P_{M_{W}Z}-\bs M_{Z,W}\bs H\bs M_{Z,W}\bs P_{M_{W}Z}\\
        &\quad-\bs P_{M_{W}Z}\bs M_{Z,W}\bs H\bs M_{Z,W}+\bs M_{Z,W}\bs H\bs M_{Z,W}\bs H\bs M_{Z,W})\bs Z\bs\Upsilon\bs a\\
        &=\frac{1}{n}\bs a'\bs\Upsilon'\bs Z'\bs P_{M_{W}Z}\bs Z\bs\Upsilon\bs a\\
        &\leq\frac{1}{n}\bs a'\bs\Upsilon'\bs Z'\bs Z\bs\Upsilon\bs a\leq\bar{C},
    \end{split}
\end{equation}
by \cref{assit:ZZ}.

For the fourth term write $\bs u_{a}=\bs\eta\bs D_{\mu}^{-1}\bs a$ and $\bs u_{b}=\bs\eta\bs D_{\mu}^{-1}\bs b$
\begin{equation}
    \begin{split}
        \E([\bs a'\bs D_{\mu}^{-1}\bs\eta'\tilde{\bs P}\bs\eta\bs D_{\mu}^{-1}\bs b]^2|\bs Z)&=\E(\bs u_{a}'\tilde{\bs P}\bs u_{ b}\bs u_{b}'\tilde{\bs P}\bs u_{a}|\bs Z)\\
        &=\E(\sum_{\substack{g\neq h\\i\neq j}}\bs u_{a,[g]}'\tilde{\bs P}_{[g,h]}\bs u_{b,[h]}\bs u_{b,[i]}'\tilde{\bs P}_{[i,j]}\bs u_{a,[j]}|\bs Z),
    \end{split}
\end{equation}
which has non-zero expectation only when $g=i\neq h=j$ or $g=j\neq h=i$. Write for the first case
\begin{equation}
    \begin{split}
        &\E(\sum_{\substack{g\neq h}}\bs u_{a,[g]}'\tilde{\bs P}_{[g,h]}\bs u_{b,[h]}\bs u_{b,[g]}'\tilde{\bs P}_{[g,h]}\bs u_{a,[h]}|\bs Z)\\
        &\leq\sum_{g\neq h}[\E(\bs u_{a,[g]}'\tilde{\bs P}_{[g,h]}\bs u_{b,[h]}\bs u_{b,[h]}'\tilde{\bs P}_{[h,g]}\bs u_{a,[g]}|\bs Z)\E(\bs u_{a,[h]}'\tilde{\bs P}_{[h,g]}\bs u_{b,[g]}\bs u_{b,[g]}'\tilde{\bs P}_{[g,h]}\bs u_{a,[h]}|\bs Z)]^{1/2}\\
        &\leq\sum_{g\neq h}[\lambdamax(\bs\Sigma_{u_{b},h})\E(\bs u_{a,[g]}'\bs P_{[g,h]}\tilde{\bs P}_{[h,g]}\bs u_{a,[g]}|\bs Z)\lambdamax(\bs\Sigma_{u_{b},g})\E(\bs u_{a,[h]}'\tilde{\bs P}_{[h,g]}\tilde{\bs P}_{[g,h]}\bs u_{a,[h]}|\bs Z)]^{1/2}\\
        &\leq\sum_{g\neq h}[\lambdamax(\bs\Sigma_{u_{b},h})\lambdamax(\bs\Sigma_{u_{a},g})\lambdamax(\bs\Sigma_{u_{b},g})\lambdamax(\bs\Sigma_{u_{a},h})\tr(\tilde{\bs P}_{[g,h]}\tilde{\bs P}_{[h,g]})^2]^{1/2}\\
        &\leq\sum_{g\neq h}\frac{C\nmax^2}{\mu_{\min}^2}\tr(\tilde{\bs P}_{[g,h]}\tilde{\bs P}_{[h,g]})\\
        &\leq\frac{C\nmax^2k}{\mu_{\min}^2},
    \end{split}
\end{equation}
where $\bs\Sigma_{u_{a},g}=\E(\bs u_{a,[g]}\bs u_{a,[g]}'|\bs Z)\leq n_g\bar{C}/\mu_{\min}^2$ and similarly for the other matrices. For the last line we used that
\begin{equation}\label{eq:tr PP}
    \begin{split}
        \sum_{g\neq h}\tr(\tilde{\bs P}_{[g,h]}\tilde{\bs P}_{[h,g]})&=\sum_{g=1}^G\tr([\tilde{\bs P}\tilde{\bs P}]_{[g,g]})-\sum_{g=1}^G\tr(\tilde{\bs P}_{[g,g]}\tilde{\bs P}_{[g,g]})\leq\sum_{g=1}^G\tr([\tilde{\bs P}\tilde{\bs P}]_{[g,g]})=\tr(\tilde{\bs P}\tilde{\bs P}),
    \end{split}
\end{equation}
which by the triangle inequality on the Frobenius norm and the definition of $\tilde{\bs P}$ is
\begin{equation}
    \begin{split}
        \tr(\tilde{\bs P}\tilde{\bs P})\leq C[\tr(\bs P_{M_WZ})+\tr(\bs M_{W,Z}\bs H\bs M_{W,Z}\bs H\bs M_{W,Z})]\leq \bar{C}k,
    \end{split}
\end{equation}
by \cref{assit:trHH}.

Consequently, $\bs D_{\mu}^{-1}\bs X'\tilde{\bs P}\bs X\bs D_{\mu}^{-1}=O_p(1)$.
\end{proof}

\begin{lemma}\label{lem:XPeps}
    Under \cref{ass:rate,ass:moments,ass:asymptotic} $\bs D_{\mu}^{-1}\bs X'\tilde{\bs P}\bs\varepsilon=O_p(\max\{\nmax,\nmax\sqrt{k}/\mumin\})$
\end{lemma}

\begin{proof}
We have
\begin{equation}\label{eq:XPeps}
    \bs D_{\mu}^{-1}\bs X'\tilde{\bs P}\bs\varepsilon=\bs\Upsilon'\bs Z'\tilde{\bs P}\bs\varepsilon/\sqrt{n}+\bs D_{\mu}^{-1}\bs\eta'\tilde{\bs P}\bs\varepsilon
\end{equation}
For the first term write with $\bs a$ as before
\begin{equation}
    \begin{split}
        \E([\bs a'\bs\Upsilon'\bs Z'\tilde{\bs P}\bs\varepsilon/\sqrt{n}]^2|\bs Z)&=\frac{1}{n}\E(\bs a'\bs\Upsilon'\bs Z'\tilde{\bs P}\bs\varepsilon\bs\varepsilon'\tilde{\bs P}\bs Z\bs\Upsilon\bs a|\bs Z)\\ 
        &\leq\frac{\bar{C}\nmax}{n}\bs a'\bs\Upsilon'\bs Z'\tilde{\bs P}\tilde{\bs P}\bs Z\bs\Upsilon\bs a,
    \end{split}
\end{equation}
by \cref{ass:moments}. Hence, using the same arguments as for \cref{eq:ZPPZ}, the first term in \cref{eq:XPeps} is $O_p(\nmax)$.

For the second term in \cref{eq:XPeps} define with slight abuse of notation and with $\bs u_{a}$ as before
\begin{equation}
    \begin{split}
        \bs v_{a,[g]}=\begin{bmatrix}
            \bs\varepsilon_{[g]}\\\bs u_{a,[g]}
        \end{bmatrix} \text{ and } \tilde{\bs P}_{[g,h]}^*=\begin{bmatrix}
            \bs 0 & \tilde{\bs P}_{[g,h]}\\
            \tilde{\bs P}_{[g,h]} &\bs 0
        \end{bmatrix}.
    \end{split}
\end{equation} 
Then we can write
\begin{equation}
    \begin{split}
        \E([\bs a'\bs D_{\mu}^{-1}\bs\eta'\tilde{\bs P}\bs\varepsilon]^2|\bs Z)
        &=\E(\bs u_{a}'\tilde{\bs P}\bs\varepsilon\bs\varepsilon'\tilde{\bs P}\bs u_{a}|\bs Z)\\
        &=\sum_{\substack{g\neq h\\i\neq j}}\E(\bs u_{a,[g]}'\tilde{\bs P}_{[g,h]}\bs\varepsilon_{[h]}\bs\varepsilon_{[i]}'\tilde{\bs P}_{[i,j]}\bs u_{[j]}|\bs Z)\\
        &=\sum_{g\neq h}\E(\bs u_{a,[g]}'\tilde{\bs P}_{[g,h]}\bs\varepsilon_{[h]}\bs\varepsilon_{[g]}'\tilde{\bs P}_{[i,j]}\bs u_{[h]}+\bs u_{a,[g]}'\tilde{\bs P}_{[g,h]}\bs\varepsilon_{[h]}\bs\varepsilon_{[h]}'\tilde{\bs P}_{[i,j]}\bs u_{[g]}|\bs Z)\\
        &=\frac{1}{2}\sum_{g\neq h}\E(\bs u_{a,[g]}'\tilde{\bs P}_{[g,h]}\bs\varepsilon_{[h]}\bs\varepsilon_{[g]}'\tilde{\bs P}_{[g,h]}\bs u_{[h]}+\bs u_{a,[g]}'\tilde{\bs P}_{[g,h]}\bs\varepsilon_{[h]}\bs\varepsilon_{[h]}'\tilde{\bs P}_{[h,g]}\bs u_{[g]}\\
        &\quad+\bs u_{a,[h]}'\tilde{\bs P}_{[g,h]}\bs\varepsilon_{[g]}\bs\varepsilon_{[h]}'\tilde{\bs P}_{[h,g]}\bs u_{[g]}+\bs u_{a,[h]}'\tilde{\bs P}_{[h,g]}\bs\varepsilon_{[g]}\bs\varepsilon_{[g]}'\tilde{\bs P}_{[g,h]}\bs u_{[h]}|\bs Z)\\
        &=\frac{1}{2}\sum_{g\neq h}\E(\bs v_{a,[g]}'\tilde{\bs P}_{[g,h]}^*\bs v_{a,[h]}\bs v_{a,[h]}'\tilde{\bs P}_{[h,g]}^*\bs v_{a,[g]}|\bs Z).
    \end{split}
\end{equation}

Now use \cref{ass:moments} to write
\begin{equation}
    \begin{split}
        \frac{1}{2}\sum_{g\neq h}\E(\bs v_{a,[g]}'\tilde{\bs P}_{[g,h]}^*\bs v_{a,[h]}\bs v_{a,[h]}'\tilde{\bs P}_{[h,g]}^*\bs v_{a,[g]}|\bs Z)&\leq \bar{C}\nmax^2\sum_{g\neq h}\tr(\tilde{\bs P}_{[g,h]}^*\tilde{\bs P}_{[h,g]}^*)/\mumin^2\\
        &=\bar{C}\nmax^2\sum_{g\neq h}\tr(\tilde{\bs P}_{[g,h]}\tilde{\bs P}_{[h,g]})/\mumin^2\leq \bar{C}\nmax^2k/\mumin^2,
    \end{split}
\end{equation}
by the definition of $\tilde{\bs P}^*_{[g,h]}$ and the same steps as in \cref{eq:tr PP}. We therefore have by Markov's inequality that $\bs D_{\mu}^{-1}\bs\eta'\tilde{\bs P}\bs\varepsilon=O_p(\nmax\sqrt{k}/\mumin)$ and hence $\bs D_{\mu}^{-1}\bs X'\tilde{\bs P}\bs\varepsilon=O_p(\max\{\nmax,\nmax\sqrt{k}/\mumin\})$.
\end{proof}

\begin{theorem}
    Under \cref{ass:rate,ass:moments,ass:asymptotic} $\|\hat{\bs\beta}_{FECJ}-\bs\beta\|\toprob 0$.
\end{theorem}

\begin{proof}
    By submultiplicativity of the norm we can write $\|\hat{\bs\beta}-\bs\beta_0\|\leq\|\bs D_{\mu}/\mumin\|\|\hat{\bs\beta}-\bs\beta_0\|\leq\|\bs D_{\mu}(\hat{\bs\beta}-\bs\beta_0)/\mumin\|=\|(\bs D_{\mu}^{-1}\bs X'\bs P\bs X\bs D_{\mu}^{-1})^{-1}\bs D_{\mu}^{-1}\bs X'\bs P\bs\varepsilon/\mumin\|$. The inverse is finite by \cref{assit:Pmwz}, such that $\|(\bs D_{\mu}^{-1}\bs X'\bs P\bs X\bs D_{\mu}^{-1})^{-1}\bs D_{\mu}^{-1}\bs X'\bs P\bs\varepsilon/\mumin\|\toprob 0$ due to \cref{lem:XPX,lem:XPeps} and \cref{ass:rate}.
\end{proof}

\section{Variance estimation}
\subsection{MD CJIVE}\label{app:variance MD CJIVE}
For ease of exposition, we present the variance estimator for the case the data are clustered in two dimensions. The results can easily be generalised to more dimensions. First consider the variance of $\bs X'\dddot{\bs P}_{Z}\bs\varepsilon$.
\begin{equation}
    \begin{split}
        \var(\bs X'\dddot{\bs P}_{Z}\bs\varepsilon)=\E(\bs X'\dddot{\bs P}_{Z}\bs\varepsilon\bs\varepsilon'\dddot{\bs P}_{Z}\bs X)=\E(\sum_{j,k=1}^n\sum_{i\notin [G(j)]}\sum_{l\notin[G(k)]}X_{i}P_{Z,ij}\varepsilon_{j}\varepsilon_{k}P_{Z,kl}X_{l}).
    \end{split}
\end{equation}
Only the terms for which either $i\in[G(k)]$ and $j\in[G(l)]$ or $j\in[G(k)]$ have expectation non-zero. Hence, 
\begin{equation}\label{eq:var MD JIVE}
    \begin{split}
        &\E(\sum_{j,k=1}^n\sum_{i\notin [G(j)]}\sum_{l\notin[G(k)]}X_{i}P_{Z,ij}\varepsilon_{j}\varepsilon_{k}P_{Z,kl}X_{l}|\bs Z)\\
        &=\E(\sum_{j,k=1}^n\sum_{i\in\{[G(k)]\setminus[G(j)]}\sum_{l\in\{[G(j)]\setminus[G(k)]}X_{i}P_{Z,ij}\varepsilon_{j}\varepsilon_{k}P_{Z,kl}X_{l}\\
        &\quad+\bs X'\dddot{\bs P}_{Z}[\bs B^{(1)}_{\varepsilon\varepsilon'}+\bs B^{(2)}_{\varepsilon\varepsilon'}-\bs B^{(1,2)}_{\varepsilon\varepsilon'}]\dddot{\bs P}_{Z}\bs X|\bs Z),
    \end{split}
\end{equation}
where the first term corresponds to the first case and the second term to the second case.

We can then estimate the variance of $\hat{\beta}_{MDCJ}-\beta$ by substituting $\bs\varepsilon$ by the residuals $\hat{\bs\varepsilon}=\bs y-\bs X\hat{\beta}_{MDCJ}$ and pre and post-multiplying with $(\bs X'\dddot{\bs P}_{Z}\bs X)^{-1}$. That is
\begin{equation}
\begin{split}
\hat{V}_{MDCJ}&=(\bs X'\dddot{\bs P}_{Z}\bs X)^{-1}[\sum_{j,k=1}^n\sum_{i\in\{[G(k)]\setminus[G(j)]}\sum_{l\in\{[G(j)]\setminus[G(k)]}X_{i}P_{Z,ij}\hat{\varepsilon}_{j}\hat{\varepsilon}_{k}P_{Z,kl}X_{l}\\
        &\quad+\bs X'\dddot{\bs P}_{Z}[\bs B^{(1)}_{\hat{\varepsilon}\hat{\varepsilon}'}+\bs B^{(2)}_{\hat{\varepsilon}\hat{\varepsilon}'}-\bs B^{(1,2)}_{\hat{\varepsilon}\hat{\varepsilon}'}]\dddot{\bs P}_{Z}\bs X](\bs X'\dddot{\bs P}_{Z}\bs X)^{-1}.
\end{split}
\end{equation}

\subsection{FE CJIVE}\label{app:variance FE CJIVE}
To derive a variance estimator of $\hat{\beta}_{FECJ}-\beta$ we first consider
\begin{equation}
    \begin{split}
        \var(\bs X'\tilde{\bs P}_{Z}\bs\varepsilon|\bs Z)&=\E([\bs X'\tilde{\bs P}_{Z}\bs\varepsilon][\bs X\tilde{\bs P}_{Z}\bs\varepsilon]'|\bs Z)=\E(\sum_{g\neq h}\sum_{i\neq j}\bs X_{[g]}'\tilde{\bs P}_{Z,[g,h]}\bs\varepsilon_{[h]}\bs\varepsilon_{[i]}'\tilde{\bs P}_{Z,[i,j]}\bs X_{[j]}|\bs Z),
    \end{split}
\end{equation}
where $\tilde{\bs P}_{Z}$ is the FE CJIVE projection matrix. This expectation has nonzero expectation only when $g\neq h=i\neq j$ or when $g=i\neq h=j$.
\begin{equation}
    \begin{split}
        &\E(\sum_{g\neq h}\sum_{i\neq j}\bs X_{[g]}'\tilde{\bs P}_{Z,[g,h]}\bs\varepsilon_{[h]}\bs\varepsilon_{[i]}'\tilde{\bs P}_{Z,[i,j]}\bs X_{[j]}|\bs Z)\\
        &=\E(\sum_{g,j=1}^G\sum_{g\neq h\neq j}\bs X_{[g]}'\tilde{\bs P}_{Z,[g,h]}\bs\varepsilon_{[h]}\bs\varepsilon_{[h]}'\tilde{\bs P}_{Z,[h,j]}\bs X_{[j]}+\sum_{g\neq h}\bs X_{[g]}'\tilde{\bs P}_{Z,[g,h]}\bs\varepsilon_{[h]}\bs\varepsilon_{[g]}'\tilde{\bs P}_{Z,[g,h]}\bs X_{[h]}|\bs Z).
    \end{split}
\end{equation}
Substituting the $\bs\varepsilon$ with the residuals $\hat{\bs\varepsilon}=\bs M_{W,Z}(\bs y-\bs X\hat{\beta}_{FECJ})$ and adding the quadratic form with $(\bs X'\tilde{\bs P}_{Z}\bs X)^{-1}$, we can write the variance estimator for $\hat{\beta}_{FECJ}-\beta$ as
\begin{equation}\label{eq:Vhat FE CJIVE}
    \begin{split}
        \hat{V}_{FECJ}=(\bs X'\tilde{\bs P}_{Z}\bs X)^{-1}[\bs X'\tilde{\bs P}_{Z}\bs B_{\hat{\varepsilon}\hat{\varepsilon}'}\tilde{\bs P}_{Z}\bs X+\sum_{g\neq h}\bs X_{[g]}'\tilde{\bs P}_{Z,[g,h]}\hat{\bs\varepsilon}_{[h]}\hat{\bs\varepsilon}_{[g]}'\tilde{\bs P}_{Z,[g,h]}\bs X_{[h]}](\bs X'\tilde{\bs P}_{Z}\bs X)^{-1}.
    \end{split}
\end{equation}

\end{appendices}

\end{document}